\documentclass[12pt]{article}
\usepackage[authoryear]{natbib}
\usepackage{amssymb}
\usepackage{amsfonts}
\usepackage{amsmath}
\usepackage{dsfont}
\usepackage{soul}
\usepackage[nohead]{geometry}
\usepackage{graphicx}
\usepackage{amsthm}
\usepackage{color}
\usepackage{comment}
\usepackage{setspace}
\usepackage{framed}
\usepackage{enumitem}
\usepackage{todonotes}
\usepackage{tikz}
\usepackage{rotating}
\usepackage{subfigure}
\usepackage[flushleft]{threeparttable}
\usepackage{multirow}
\usepackage[colorlinks=true,citecolor=blue,linkcolor=blue]{hyperref}
\usepackage{xr}
\usepackage{bm}
\usepackage{xcolor}
\usepackage{booktabs}
\usepackage{lscape}
\usepackage{algorithm}
\usepackage{algorithmic}
\usepackage[toc,title,page]{appendix}
\usepackage{lscape}
\usepackage{pdflscape}
\usepackage{blkarray}
 \usepackage{float} 

\usepackage{natbib}

\renewcommand{\today}{\ifcase \month \or January\or February\or March\or %
	April\or May\or June\or July\or August\or September\or October\or November\or %
	December\fi, \number \year} 

\usepackage{caption} 
\allowdisplaybreaks

7

\newcommand{\bA}{\mbox{\bf A}}

\newcommand{\bE}{\mbox{\bf E}}

\newcommand{\bV}{\mbox{\bf V}}

\newcommand{\bU}{\mbox{\bf U}}

\newcommand{\bLambda}{\mbox{\boldmath $\Lambda$}}

\newcommand{\bSigma}{\mbox{\boldmath $\Sigma$}}
\newcommand{\bOmega}{\mbox{\boldmath $\Omega$}}

\newcommand{\cov}{\mathrm{cov}}

\newcommand{\tr}{\mathrm{tr}}

\newcommand{\beq}{\begin{eqnarray*}}
\newcommand{\eeq}{\end{eqnarray*}}

\usepackage{graphicx,rotating,booktabs,threeparttable}

\usepackage{adjustbox}
\usepackage{array}

\newcolumntype{R}[2]{%
	>{\adjustbox{angle=#1,lap=\width-(#2)}\bgroup}%
	l%
	<{\egroup}%
}

\newtheorem{thm}{Theorem}[section]

\newtheorem{lem}{Lemma}[section]

\newtheorem{assum}{Assumption}[section]

\numberwithin{equation}{section}
\theoremstyle{definition}

\newtheorem{remark}{Remark}[section]
\makeatletter
\def\@biblabel#1{\hspace*{-\labelsep}}
\makeatother
\numberwithin{equation}{section}

\renewcommand{\hat}{\widehat}

\renewcommand{\hat}{\widehat}

\newcommand{\bfm}[1]{\ensuremath{\mathbf{#1}}}

   \def\bA{\bfm A}

\def\be{\bfm e}   \def\bE{\bfm E}  
    \def\FF{\mathbb{F}}

   \def\bU{\bfm U}  
   \def\bV{\bfm V}

 \def\cF{{\cal  F}}

\newcommand{\bfsym}[1]{\ensuremath{\boldsymbol{#1}}}

              \def\bSigma{\bfsym \Sigma}
         \def\bLambda {\bfsym {\Lambda}}
           \def\bOmega {\bfsym {\Omega}}

\def\1{\bfsym{1}}	

\DeclareMathOperator{\Var}{Var}

\def\newpage{\vfill\eject}

\def\today{\ifcase\month\or
  January\or February\or March\or April\or May\or June\or
  July\or August\or September\or October\or November\or December\fi
  \space\number\day, \number\year}

\newdimen\biblioindent    \biblioindent=30pt

 at 8truept

\newcommand{\beqn}{\begin{eqnarray}}
  \newcommand{\eeqn}{\end{eqnarray}}
\newcommand{\beqnn}{\begin{eqnarray*}}
  \newcommand{\eeqnn}{\end{eqnarray*}}

\allowdisplaybreaks
\usepackage{verbatim}
\def\tilde{\widetilde}

\def\FF{\mathcal{F}}
\def\[{\left [}  \def\]{\right ]} \def\({\left (}  \def\){\right )}
 
\def\hat{\widehat}

\theoremstyle{definition}

 \def \Diag {\mathrm{Diag}}

\begin{document}
	\bibliographystyle{ecta}
	
	\title{Low-Rank Structured Nonparametric Prediction of Instantaneous Volatility}

	\date{July 2025}
	
	\author{
		Sung Hoon Choi\thanks{%
			Department of Economics, University of Connecticut, Storrs, CT 06269, USA. 
			E-mail: \texttt{sung\_hoon.choi@uconn.edu}.} \\ 
		\and Donggyu Kim\thanks{Department of Economics, University of California, Riverside, CA 92521, USA.
			Email: \texttt{donggyu.kim@ucr.edu}.}
		}
	\maketitle
	\pagenumbering{arabic}	
\begin{abstract}
		\onehalfspacing
Based on Itô semimartingale models, several studies have proposed methods for forecasting intraday volatility using high-frequency financial data. 
These approaches typically rely on restrictive parametric assumptions and are often vulnerable to model misspecification.
 To address this issue, we introduce a novel nonparametric prediction method for the future intraday instantaneous volatility process during trading hours, which leverages both previous days' data and the current day's observed intraday data. 
 Our approach imposes an interday-by-intraday matrix representation of the instantaneous volatility, which is decomposed into a low-rank conditional expectation component and a noise matrix.
 To predict the future conditional expected volatility vector, we exploit this low-rank structure and propose the Structural Intraday-volatility Prediction (SIP) procedure.
 We establish the asymptotic properties of the SIP estimator and demonstrate its effectiveness through an out-of-sample prediction study using real high-frequency trading data.
		\\

\end{abstract}

\newpage
	
\doublespacing

\section{Introduction}
The analysis of volatility is essential in financial econometrics and statistics, and it has broad applications in hedging, option pricing, risk management, and portfolio allocation. 
In recent years, the widespread availability of high-frequency financial data has led to the development of several nonparametric integrated volatility estimation methods \citep{ait2010high, barndorff2008designing, barndorff2011multivariate, bibinger2014estimating, christensen2010pre, fan2018robust, fan2007multi, jacod2009microstructure, xiu2010quasi, zhang2005tale, zhang2006efficient, zhang2011estimating}.
These high-frequency-based estimators have significantly improved our understanding of low-frequency (interday) market dynamics and have motivated the development of various conditional volatility models built on realized measures \citep{andersen2003modeling, corsi2009simple, hansen2012realized, kim2016unified, shephard2010realising}. 
Parallel to these developments, nonparametric procedures for estimating instantaneous (or spot)  volatility have been proposed to capture intraday volatility dynamics \citep{fan2008spot,figueroa2022kernel, foster1996continuous, kristensen2010nonparametric, mancini2015spot, todorov2019nonparametric, todorov2023bias, zu2014estimating}.
Leveraging these well-performing estimators, several studies have illustrated U-shaped or reverse J-shaped intraday volatility patterns across various markets \citep{admati1988theory, andersen1997intraday, andersen2019time, hong2000trading, li2023robust}.

In contrast to the extensive literature on predicting daily volatility processes, relatively less attention has been given to forecasting intraday volatility.
For example, \cite{engle2012forecasting} proposed a modified GARCH model in which the conditional variance is expressed as the product of daily, diurnal, and stochastic intraday components to capture intraday volatility dynamics, while \cite{zhang2024volatility} explored several nonparametric machine learning methods for forecasting intraday realized volatility, leveraging the commonality observed in intraday volatility patterns.
Recently, \cite{TIP-PCA} considered the interday-by-intraday instantaneous volatility matrix process and proposed a semiparametric prediction procedure for the one-day-ahead instantaneous volatility process by leveraging both interday and intraday volatility dynamics.
However, existing methods for forecasting intraday volatility suffer from notable limitations. 
For example, many methods rely on strong distributional assumptions inherent in parametric models and often fail to incorporate interday dynamics effectively. 
 In practice, practitioners are often interested in forecasting the remaining future intraday volatility process conditional on the observed intraday data up to the present time.
Therefore, it is both practically relevant and methodologically important to develop a prediction procedure that adapts to current intraday information and is robust to potential model misspecification.


This paper introduces a novel nonparametric prediction approach based on a low-rank structure for forecasting the remaining future instantaneous volatility process during intraday trading.
Specifically, we represent the entire instantaneous volatility process as a matrix, where each row corresponds to a day and each column represents an intraday time point.
We impose a low-rank plus noise structure on this matrix. 
The parameter of interest is the instantaneous volatility process after the current intraday time.
That is, on the current day (denoted as the $D$th day), only the initial portion of the row is observed, while the remaining part is the target of prediction.
To forecast this remaining portion, we exploit the low-rank matrix structure. 
For example, the future instantaneous volatility vector can be represented by the singular components of the observed part of the low-rank matrix. 
We use this structure to construct a nonparametric prediction of the future volatility vector.
This method is called the Structural Intraday-volatility Prediction (SIP) procedure.
It is worth noting that as long as the interday-by-intraday instantaneous volatility matrix exhibits a low-rank structure, the SIP method can effectively predict the remaining future volatility vector without relying on any specific parametric dynamics such as autoregressive models.
Therefore, the SIP procedure is robust to model misspecification arising from incorrect parametric assumptions.
Furthermore, we derive the convergence rate of the predicted instantaneous volatility under the SIP framework and assess its empirical performance via an out-of-sample prediction study using high-frequency financial data.
The SIP method consistently outperforms alternative approaches in terms of both out-of-sample prediction accuracy and Value-at-Risk (VaR) forecasting performance.


The remainder of the paper is organized as follows.
In Section \ref{setup}, we introduce the model and propose the SIP prediction procedure.
Section \ref{asymp} presents the asymptotic properties of the SIP estimator.
In Section \ref{simulation}, we evaluate the finite sample performance of the proposed method through simulation.
Section \ref{empiric} applies the proposed method to real high-frequency financial data to forecast the remaining future instantaneous volatility process during trading hours.
Finally, the conclusion is provided in Section \ref{conclusion}.
All technical proofs are provided in the Appendix \ref{proofs}.

\section{Model Setup and Estimation Procedure} \label{setup}
Throughout this paper, we denote by $\|\bA\|_{F}$, $\|\bA\|_{2}$ (or $\|\bA\|$ for short), $\|\bA\|_{1}$, $\|\bA\|_{\infty}$, and $\|\bA\|_{\max}$ the Frobenius norm,  operator norm, $l_{1}$-norm, $l_{\infty}$-norm, and element-wise norm, which are defined, respectively, as $\|\bA\|_{F} = \tr^{1/2}(\bA^{\top}\bA)$, $\|\bA\|_{2} = \lambda_{\max}^{1/2}(\bA^{\top}\bA)$, $\|\bA\|_{1} = \max_{j}\sum_{i}|a_{ij}|$, $\|\bA\|_{\infty} = \max_{i}\sum_{j}|a_{ij}|$, and $\|\bA\|_{\max} = \max_{i,j}|a_{ij}|$. 
When $\bA$ is a vector, the maximum norm is denoted as $\|\bA\|_{\infty}=\max_{i}|a_{i}|$, and both $\|\bA\|$ and $\|\bA\|_{F}$ are equal to the Euclidean norm.

\subsection{A Model Setup} \label{model}
We consider the following jump diffusion process: for the $i$-th day and intraday time $t \in [0,1],$
\begin{equation}\label{diffusion-def}
	dX_{i,t}= \mu_{i,t} dt + \sigma_{i,t}  dB_{i,t} + J_{i,t} d P_{i,t}, 
\end{equation}
where $X_{i,t}$ is the log price of an asset, $\mu_{i,t}$ is a drift process,  $B_{i,t}$ is a one-dimensional standard Brownian motion,  $J_{i,t}$ is the jump size, and $P_{i,t}$ is the Poisson process with the intensity $\mu_{J}$.
For a given intraday time sequence, for each $i =1,\dots, D$ and $j=1,\dots, n$, we denote the instantaneous volatility process as $c_{i,j} := \sigma_{i,t_{j}}^2$,  where $0< t_{1} < \cdots < t_{n} =1$.
Then, we assume that the discrete-time instantaneous volatility process satisfies the following low-rank structure: 
\begin{equation}\label{Sigma}
	\bSigma_{D, n} = (c_{i,j}) _{D \times n}  = \bU \bLambda \bV^{\top} + \bE \equiv \bA + \bE,
\end{equation}
where $\bU=  (u_{i,k} ) _{i=1,\ldots, D, k=1,\ldots,r}$ is the left singular vector matrix, $\bV = (v_{j,k}) _{j=1, \ldots, n, k=1,\ldots, r}$ is the right singular vector matrix, $\bLambda= \Diag (\lambda_1, \ldots, \lambda_r)$ is the singular value matrix, and $\bE = (e_{i,j})_{D\times n}$ is the random noise matrix.  
We only impose the above low-rank structure in \eqref{Sigma} to predict future remaining intraday volatility.
 It allows for considerable modeling flexibility and can implicitly incorporate various time-series features.
For example, the left singular vector can represent interday time-series dynamics, while the right singular vector can account for intraday patterns, such as  U-shaped or reverse J-shaped \citep{TIP-PCA}.
It is worth noting that, unlike  \citet{TIP-PCA}, we do not assume any explicit dynamic or parametric form for the volatility evolution.
We note that $\bSigma_{D, n}$ is a $D \times n$ matrix, which is distinct from a covariance matrix and contains only positive elements.

In this paper, our target is to predict the remaining future instantaneous volatility vector for the $D$th day.
Specifically, we assume that trading data is available up to a fraction $\omega = \frac{n_1}{n}$ of the $D$th day. 
The objective is to predict the remaining $n_2 \times 1$ instantaneous vector, $(c_{D,n_1+1},\dots, c_{D,n})$, where $n = n_1 + n_2$.
Given  the current information $\FF_{D,n_1}$, we can predict the instantaneous volatility as follows: for $j = n_1+1, \ldots, n$,
 \begin{equation} \label{conditional exp}
 	E\left[\sigma_{D,j}^2 | \FF_{D,n_1}\right] =  \sum_{k=1}^r \lambda_k u_{D,k} v_{j,k} \,\, \text{ a.s.}
 \end{equation}
We can write the model \eqref{Sigma} in a partitioned matrix form as follows:
$$
\bSigma_{D,n} =
\begin{array}{c@{}c}
    \footnotesize  \begin{array}{cc} n_1 & n_2 \end{array} \\
    \left[
    \begin{array}{cc}
        \Sigma_{11} & \Sigma_{12} \\
        \Sigma_{21} & \Sigma_{22}
    \end{array}
    \right]& \footnotesize \begin{array}{c} D-1 \\ 1 \end{array} 
\end{array},
\bA =
\begin{array}{c@{}c}
    \footnotesize  \begin{array}{cc} n_1 & n_2 \end{array} \\
    \left[
    \begin{array}{cc}
        A_{11} & A_{12} \\
        A_{21} & A_{22}
    \end{array}
    \right]& \footnotesize \begin{array}{c} D-1 \\ 1 \end{array} 
\end{array},
\bE =
\begin{array}{c@{}c}
    \footnotesize  \begin{array}{cc} n_1 & n_2 \end{array} \\
    \left[
    \begin{array}{cc}
        E_{11} & E_{12} \\
        E_{21} & E_{22}
    \end{array}
    \right]& \footnotesize \begin{array}{c} D-1 \\ 1 \end{array}. 
\end{array}
$$
We note that  $\Sigma_{11}, \Sigma_{12}$, and $\Sigma_{21}$ are estimable using high-frequency log-price observations. 
Since $\bA$ is the low rank matrix, $A_{22}$ is given by
\begin{equation}\label{target}
    A_{22} = A_{21}(A_{11})^\dagger A_{12} = A_{21}(V_{11}\Lambda_{11}^{-1}U_{11}^{\top})A_{12},
\end{equation}
where $U_{11}$ is the $(D-1)\times r$ left singular vector matrix, $V_{11}$ is the $n_1 \times r$ right singular vector matrix, and $\Lambda_{11}$ is the $r\times r$ singular value matrix of $A_{11}$.
Since $A_{11}$, $A_{21}$, and $A_{12}$ are quantities related to the observed period, using this relationship \eqref{target}, we can predict $A_{22}$ without a specific parametric modeling assumption. 
Specifically, as long as we can estimate $A_{11}$, $A_{21}$, and $A_{12}$ well using the available data, $A_{22}$ can be estimated using the plug-in scheme. 
We note that if we treat the remaining intraday volatility vector $A_{22}$ as missing components, this problem becomes similar to the matrix completion \citep{cai2010singular,  candes2012exact, candes2010matrix}. 
Especially in this problem, the missing has a determined structure as in \citet{bai2021matrix, cai2016structured,fan2019structured, yan2024entrywise}.
From this point of view, this paper extends the matrix completion concept to predict time series patterns.

 Due to the imperfections in the trading process, practitioners cannot directly observe the true underlying log-stock price $X_{i,t}$ in \eqref{diffusion-def} \citep{ait2009high}. 
 To account for this, we assume that the high-frequency intraday observations $X_{i,t_s}, s=1, \dots, m,$  are contaminated by microstructure noise as follows:
	\begin{equation} 	\label{def-obervation}
	Y _{i,t_s} = X_{i,t_s} + e_{i,t_s}, \quad i=1,\ldots,D, s=1,\ldots,m,
	\end{equation}
where $e_{i,t_s}$ is the microstructure noise process with a mean of zero and a variance of $\eta_{ii}$. 


\subsection{Structural Intraday-volatility Prediction} \label{estimation procedure}

To predict the remaining future instantaneous volatility vector $A_{22}$ nonparametrically, we use the low-rank structural equation \eqref{target}. 
Since $A_{11}$, $A_{12}$, and $A_{21}$ are latent low-rank matrices, we need to estimate them using a well-performing instantaneous volatility estimator based on the observed high-frequency data. 
For example, we need a large dimension to estimate the latent low-rank matrix and use the blessing of dimensionality \citep{li2018embracing}.
To do this, we choose the large block matrices such as $(\Sigma_{11}, \Sigma_{12})$ and $(\Sigma_{11} ^{\top}, \Sigma_{21} ^{\top})$. 
On the other hand, we employ a well-performing nonparametric instantaneous volatility estimator to estimate the instantaneous volatility. Details can be found in Remark \ref{remark_spot} below.
The specific estimation procedure is as follows:
\begin{enumerate}
  	\item Using high-frequency log-price observations, we compute the instantaneous volatility estimators $\hat{c}_{i,j}$ up to the $D$th day $n_1$ intraday time point (see Remark \ref{remark_spot}).
    Then, we define the following submatrices: $\hat{\Sigma}_{11} = (\hat{c}_{i,j})_{(D-1)\times n_1}, \hat{\Sigma}_{12} = (\hat{c}_{i,j})_{(D-1)\times n_2}$, and $\hat{\Sigma}_{21} = (\hat{c}_{i,j})_{1\times n_1}$, which are the initial estimated matrices for $\Sigma_{11}$,$\Sigma_{12}$, and $\Sigma_{21}$, respectively.

    \item We conduct the singular value decomposition to the large submatrices $\hat{\Sigma}_{1\bullet} \equiv [\hat{\Sigma}_{11} \: \hat{\Sigma}_{12}]$ and $\hat{\Sigma}_{\bullet 1} \equiv [\hat{\Sigma}_{11}^{\top} \: \hat{\Sigma}_{21}^{\top}]^{\top}$.
    The columns of $\hat{U}_{11}$ are defined as the $r$ leading left singular vector of $\hat{\Sigma}_{1\bullet}$, and the columns of $\hat{V}_{11}$ are defined as the $r$ leading right singular vector of $\hat{\Sigma}_{\bullet 1}$.
    Note that $\hat{U}_{11}$ and $\hat{V}_{11}$ are estimators for the orthonormal basis of column vectors of $U_{11}$ and $V_{11}$ defined in \eqref{target}, respectively.
  	
   \item Finally, we estimate the conditional expectation of the remaining instantaneous volatility vector on the $D$th day, $E[(c_{D,n_1 + 1},\dots,c_{D,n}) | \FF_{D,n_{1}}]$, by
   $$
   \tilde{\Sigma}_{22} := (\tilde{c}_{D,n_{1}+ 1},\dots, \tilde{c}_{D,n}) = \hat{\Sigma}_{21}\hat{V}_{11}(\hat{U}_{11}^{\top}\hat{\Sigma}_{11}\hat{V}_{11})^{-1}\hat{U}_{11}^{\top}\hat{\Sigma}_{12}.
   $$

  \end{enumerate}
  \begin{remark} \label{remark_spot}
     For the SIP procedure, we first need to estimate the instantaneous volatilities using the observed data. 
     In the high-frequency literature, several nonparametric methods for instantaneous volatility estimation have been proposed \citep{fan2008spot, figueroa2022kernel, foster1996continuous, kristensen2010nonparametric, mancini2015spot, todorov2019nonparametric, todorov2023bias, zu2014estimating}. 
    Any well-performing estimator of the instantaneous volatility, denoted by  $\hat{c}_{i,j}$, that satisfies Assumption \ref{assum_spotvol}(iii) can be used within our framework.      
    In the numerical study, we adopt the jump-robust pre-averaging method proposed by \cite{figueroa2022kernel}, which is described explicitly in \eqref{pre-spot} in Section \ref{simulation}.
  \end{remark}
  
  \begin{remark}
      To implement the SIP procedure, we need to select the number of factors. 
      The number of latent factors, $r$, can be determined using data-driven methods \citep{ahn2013eigenvalue, bai2002determining, onatski2010determining}.
      For instance, $r$ can be chosen by maximizing the singular value ratio or the largest singular value gap, given by  $\max_{k\leq r_{\max}}\frac{\hat{\lambda}_{k}}{\hat{\lambda}_{k+1}}$ or $\max_{k\leq r_{\max}}(\hat{\lambda}_{k}-\hat{\lambda}_{k+1})$, where $\hat{\lambda}_{1}\geq \hat{\lambda}_{2}\geq \cdots \geq \hat{\lambda}_{r_{\max}}$ are the leading singular values of $\hat{\bSigma}_{1\bullet} \equiv (\hat{c}_{i,j})_{(D-1)\times n}$ for a predetermined maximum number of factors, $r_{\max}$.
    In this paper, we choose $r=1$ in the numerical study (Sections \ref{simulation} and \ref{empiric}), following the eigenvalue ratio method proposed by \cite{ahn2013eigenvalue}.
  \end{remark}

In summary, given the estimated instantaneous volatility process, we construct a structured matrix that includes the remaining future intraday volatility vector.
To forecast the future intraday volatility vector, we utilize the low-rank structure described above.
We refer to this procedure as the Structural Intraday-volatility Prediction (SIP) method.
The SIP method nonparametrically predicts the future instantaneous volatility vector using historical trading data along with current intraday market observations.
Crucially, by incorporating current intraday information, the prediction is formulated under a low-rank matrix assumption without relying on any specific parametric model, such as an autoregressive structure.
This feature makes the SIP procedure robust to model misspecification arising from incorrect parametric assumptions.
The numerical results in Sections \ref{simulation} and \ref{empiric} demonstrate that the SIP method effectively predicts the remaining future instantaneous volatility vector.

\section{Asymptotic Properties} \label{asymp}
In this section, we develop the asymptotic properties of the SIP estimator. 
To achieve this,  the following technical conditions are needed.
 \begin{assum} \label{assum_spotvol} ~
      \begin{itemize}
            \item[(i)] For $k \leq r$, the eigengap satisfies $|\lambda_{k+1} - \lambda_{k}| = O_{P}(\sqrt{nD})$ and $\lambda_{r+1} = 0$.
            
            \item[(ii)] For some bounded $\mu_{1}$ and $\mu_{2}$, the left and right singular vectors satisfy
            $$
            \max_{i\leq D, k\leq r}|u_{i,k}|^2 \leq \frac{\mu_1}{D}, \text{ and } \max_{j\leq n, k\leq r}|v_{j,k}|^2 \leq \frac{\mu_2}{n}.
            $$
            In addition, for each $k \leq r$, $u_{k} = (u_{1,k},\dots,u_{D,k})^{\top}$ and $v_{k} = (v_{1,k},\dots,v_{n,k})^{\top}$ are assumed to be constant vectors.
            
             \item[(iii)] For each $i\leq D$ and $j\leq n$, the estimated instantaneous volatility $\hat{c}_{i,j}$ satisfies
            $$
             \hat{c}_{i,j}-c_{i,j}  = \upsilon_{i,j} + \varsigma_{i,j},
            $$
            where $\upsilon_{i,j}$ follows the martingale difference sequence and $\varsigma_{i,j}$ is the estimation bias term such that $E(\upsilon_{i,j} | \cF_{i,t_j})=0$ a.s.
            For any $s \leq D$ and $j'\leq n$, we assume $E(  c_{i,j'}\upsilon_{i,j} | \cF_{i,t_j}) = 0$ and $E( c_{s,j} c_{s,j'}\upsilon_{i,j} | \cF_{i,t_j}) = 0$ a.s. 
            In addition, for all $i\neq s$ and any $j,j'\leq n$, we assume $E(\upsilon_{i,j}\upsilon_{s,j'})=0$.
            We further assume that $c_{i,j}$, $\upsilon_{i,j}$, and $\varsigma_{i,j}$ are sub-Gaussian random variables; that is, there exists a constant $K>0$ such that for all $i\leq D$, $j\leq n$, and all $\tau \in \mathbb{R}$,
            $$
            E[\exp(\tau \alpha_{i,j})] \leq \exp\left(\frac{K^{2}\tau^{2}}{2}\right) \quad \text{for} \quad \alpha_{i,j} \in \{c_{i,j}, m^{1/8} \upsilon_{i,j}, \rho_{m}^{-1} \varsigma_{i,j}\},
            $$
            where $\rho_m$ is the convergence rate of the bias term. 
            \item[(iv)] Let $\bOmega_{e_j,D\times D}$ denote the $D\times D$ covariance matrix of $\be_{\cdot j} = (e_{1,j},\ldots, e_{D,j})^{\top}$ for each $j = 1,\dots, n$. Define $\bOmega_{j,D\times D} = \frac{1}{n_1}U_{\bullet 1}\Lambda_{\bullet 1}^{2}U_{\bullet 1}^{\top} + \bOmega_{e_j,D\times D}$. The instantaneous volatility matrix, $\bSigma_{\bullet 1} \equiv [\Sigma_{11}^{\top} \Sigma_{21}^{\top}]^{\top}$, satisfies
            $$
            \|\bSigma_{\bullet 1}\bSigma_{\bullet 1}^{\top} - \sum_{j=1}^{n_{1}}\bOmega_{j,D\times D}\|_{\max} = O_{P}(\sqrt{n_{1}\log D}).        
            $$
            Similarly, let $\bOmega_{e_i,n\times n}$ denote the $n\times n$ covariance matrix of $\be_{i \cdot} = (e_{i,1},\ldots, e_{i,n})^{\top}$ for each $i = 1,\dots, D$. Define $\bOmega_{i, n\times n} = \frac{1}{D}\bV\bLambda^{2}\bV^{\top} + \bOmega_{e_{i},n\times n}$. The instantaneous volatility matrix, $\bSigma$, satisfies
            $$
            \|\bSigma^{\top}\bSigma - \sum_{i=1}^{D}\bOmega_{i,n\times n}\|_{\max} = O_{P}(\sqrt{D\log n}).     
            $$           
            \item[(v)]  The covariance matrices $\bOmega_{e_{j},D\times D} =\cov(\be_{\cdot j})$ for each $j = 1,\ldots,n$ and $\bOmega_{e_{i},n\times n} =\cov(\be_{i \cdot})$ for each $i = 1,\ldots, D$, where $\be_{\cdot j} = (e_{1,j},\ldots, e_{D,j})^{\top}$ and $\be_{i \cdot} = (e_{i,1},\ldots, e_{i,n})^{\top}$, satisfy
          \begin{align*}
            \max_{j\leq n}\|\bOmega_{e_{j},D\times D}\|_{1} = O(\varphi_D) \quad \text{and} \quad 
            \max_{i \leq D}\|\bOmega_{e_{i},n\times n}\|_{1} = O(\varphi_n), 
          \end{align*}
          where $\varphi_{D} = o(D)$ and $\varphi_{n} = o(n)$, respectively.
         In addition, $\|\bE\| = O_{P}(\max(\sqrt{D},\sqrt{n}))$.

          \item[(vi)] There exist positive constants $C_{1}$ and $C_{2}$ such that, for each $k\leq r$, $j, l\leq n$, and $i,q \leq D$, 
        \begin{align*}            &|\cov(\sum_{i=1}^{D}u_{i,k}e_{i,j},\sum_{i=1}^{D}u_{i,k}e_{i,l})|\leq C_{1}\varphi_{D}\rho_{1}(|j-l|),\quad \text{where} \quad \sum_{h=1}^{\infty}\rho_{1}(h) <\infty, \\
        &|\cov(\sum_{j=1}^{n}v_{j,k}e_{i,j},\sum_{j=1}^{n}v_{j,k}e_{q,j})|\leq C_{2}\varphi_{n}\rho_{2}(|i-q|),\quad \text{where} \quad  \sum_{h=1}^{\infty}\rho_{2}(h) <\infty.
        \end{align*}
         In addition, for each $k \leq r$, $j \leq  n$, and $i \leq D$, we have $E[(\sum_{i=1}^{D}u_{i,k}e_{i,j})^4] = O(\varphi_{D}^2)$ and $E[(\sum_{j=1}^{n}v_{j,k}e_{i,j})^4] = O(\varphi_{n}^2)$.
        \end{itemize}
  \end{assum}

\begin{remark}
Assumption \ref{assum_spotvol}(i) ensures a pervasive assumption, which is essential for the analysis of low-rank matrix structures (see, e.g., \citealp{candes2010matrix, cho2017asymptotic, fan2018large}). 
Given that our instantaneous volatility matrix is of size $D \times n$, this eigengap condition suggests that the leading eigenvalues are associated with the low-rank structure scale on the order of $\sqrt{nD}$.
Assumption \ref{assum_spotvol}(ii) imposes the conventional incoherence condition, which ensures effective entrywise control in low-rank matrix estimation.
Assumption \ref{assum_spotvol}(iii) is the sub-Gaussian conditions that can be justified under some mild assumptions on the process $X$, the microstructure noise, and the kernel function \citep{figueroa2022kernel,kim2016asymptotic, kim2016sparse}. 
Moreover, this condition also holds for heavy-tailed data given the bounded fourth moments condition using a truncated estimation scheme \citep{fan2018robust, shin2023adaptive}.
In practice, estimation of instantaneous volatility at time $t$ typically uses data up to and including time $t$, leading to the martingale difference property $E(\upsilon_{i,j} | \cF_{i,t_j})=0$ almost surely. 
The resulting instantaneous volatility estimator is asymptotically unbiased, and the bias vanishes at a rate faster than $m^{-\frac{1}{8}}$, which is the optimal convergence rate of the instantaneous volatility estimator in the presence of microstructure noise.
To justify the effectiveness of the smoothing technique, we require an uncorrelatedness condition:  $E( c_{i,j'}\upsilon_{i,j} | \cF_{i,t_j})= E( c_{s,j} c_{s,j'}\upsilon_{i,j} | \cF_{i,t_j}) = 0$ almost surely. 
This condition holds when the spot volatility process is adapted to $\cF_{i,t_j}$ and the noise is independent. 
Such conditions are mild and are commonly satisfied in standard time series settings like ARMA models. 
Assumption \ref{assum_spotvol}(iv) is the element-wise convergence condition for analyzing large matrix inference.
Under the sub-Gaussian condition and mixing time dependency, we can easily obtain this condition (see \citealp{fan2018large, fan2018eigenvector, vershynin2010introduction, wang2017asymptotics}). 
Assumption \ref{assum_spotvol}(v) imposes the conventional condition on the idiosyncratic covariance matrix, such as sparsity, which is commonly assumed in empirical applications \citep{boivin2006more, fan2016incorporating}.
We note that our framework allows for heterogeneity across both the intraday and interday dimensions.
Assumption \ref{assum_spotvol}(vi) requires weak temporal dependence in both the linear and quadratic forms of the projected error terms, as well as bounded fourth moments.
\end{remark}

We obtain the following elementwise convergence rate of the predicted instantaneous volatility using the SIP method.
    \begin{thm} \label{main_thm}
    Suppose that $r$ is fixed, $\log D = o(n_{1})$, $\log n = o(D)$, $\varphi_{D} \sqrt{n_{1}} \leq D$, $\varphi_{n}\sqrt{D} \leq n_{1}$ and Assumption \ref{assum_spotvol} hold.
    As $D,n_1,n,m \rightarrow \infty$, we have
        \begin{align*}
            &\max_{j\leq n_2}\left|\tilde{c}_{D,n_1+j} - E[c_{D,n_1+j} | \FF_{D,n_1}]\right| = O_{P}\Bigg( \rho_{m} + \frac{1}{m^{\frac{1}{4}}}   + \frac{\varphi_{D}}{D} + \sqrt{\frac{\log D}{n_1}}+    \frac{\varphi_{n}}{n_{1}}+ \sqrt{\frac{\log n}{D}}\Bigg).   
        \end{align*}    
    \end{thm}

Theorem \ref{main_thm} shows that the proposed SIP estimator consistently predicts the future instantaneous volatility process.
The convergence rate of the SIP estimator is given by $m^{-1/4} +    \frac{1}{\sqrt{D}} + \frac{1}{\sqrt{n_1}}$ up to the log order and the sparsity level, when $\rho_m < m^{-1/4}$.
Each term in this rate corresponds to a distinct component of the estimation challenge:
 $D^{-1/2}$ reflects the statistical cost of learning the interday (daily) time-series dynamics;
$n_1^{-1/2} $ captures the cost of estimating intraday volatility patterns;
 $ m^{-1/4}$ arises from estimating the unobserved instantaneous volatility using high-frequency returns.
Notably, while the optimal nonparametric convergence rate for estimating instantaneous volatility is $m^{-1/8}$, the SIP method achieves a faster rate of $m^{-1/4}$. 
This improvement stems from the low-rank structure, which allows the SIP estimator to benefit from a law-of-large-numbers-type averaging across the volatility surface.
In addition, the term $ \rho_m$ arises from the bias component in the spot volatility estimator.
Specifically, the bias originates from the drift term, which typically decays at the rate  $m^{-1}$; however, due to the subsampling approach used to mitigate microstructure noise, the effective convergence rate  $\rho_m$  is generally of order  $m^{-1/2}$. 



\section{Simulation Study} \label{simulation}
In this section, we examined the finite sample performance of the proposed SIP method through simulations.
We generated high-frequency observations that incorporate both interday HAR model feature and the intraday periodic pattern \citep{TIP-PCA} as follows: for $i= 1,\dots, D$, $s=0,\dots, m$, and $t_{s} = s/m$, 
    \begin{align*}
        &Y_{i,t_{s}} = X_{i,t_{s}} + e_{i,t_{s}},\\
        &dX_{i,t} = (\mu - \sigma_{i,t}^2/2)dt + \sigma_{i,t} dB_{i,t} + J_{i,t}dP_{i,t},\\
        &\sigma_{i,t_{s}}^2 = \tilde{\sigma}_{i}^2 h(t_{s}) + \varepsilon_{i,t_{s}},
    \end{align*}
    where we set microstructure noise as $e_{i,t_{s}} \sim \mathcal{N}(0,0.0005^2)$ and the initial value as $X_{1,0} = 1$; $B_{i,t}$ is a standard Brownian motion; for the jump part, we set $J_{i,t} \sim \mathcal{N}(-0.01, 0.02^2)$ and $P_{i,t+\Delta} - P_{i,t} \sim \text{Poisson}(36\Delta/252)$; $\tilde{\sigma}_{i} = b_{0} + b_{1}\tilde{\sigma}_{i-1} + b_{2}\frac{1}{5}\sum_{s=1}^{5}\tilde{\sigma}_{i-s} + b_{3}\frac{1}{22}\sum_{s=1}^{22}\tilde{\sigma}_{i-s} + \zeta_{i}$,  $\zeta_{i} \sim \mathcal{N}(0,1)$,  $h(t_{s}) = \gamma_{0} + \gamma_{1}(t_{s}-0.6)^2$, and $\varepsilon_{i,t_{s}} = q(t_{s})\xi_{i,t_{s}}$, where $q(t_{s})^{2} = 0.1+0.5(2t_{s}-1)^{2}$ and $\xi_{i,t_{s}} \sim \mathcal{N}(0,0.01^2)$.
The model parameters were set to be
    \begin{align*}
    \mu = 0.05/252, \, \gamma_{0} = 0.04/252, \, \gamma_{1} = 0.5/252, \\
    b_{0} = 0.5, \, b_{1} =0.372, \, b_{2} = 0.343, \, b_{3} = 0.224.
    \end{align*}
The normalized parameter values above imply the daily time unit, and we adapted the estimated coefficients studied in \cite{corsi2009simple} to generate $\tilde{\sigma}_{i}$.
We note that in each simulation, the instantaneous volatility process was generated until all instantaneous volatility values were positive, based on the data-generating process described above.
We set $m =$ 23,400, which indicates that the data are observed every second over a period of 6.5 trading hours per day.

For each simulation, we used the jump robust pre-averaging method \citep{figueroa2022kernel} to estimate the instantaneous volatility, $c_{i,\tau}=\sigma_{i,\frac{\tau}{n}}^2$, at a frequency of every 5 minutes (i.e., $n=78$) for each $i$-th day as follows: for $\tau = 1, \dots, n$,
\begin{equation}\label{pre-spot}
\hat{c}_{i,\tau} = \frac{1}{\phi_{k_{n}}(g)}\sum_{s=1}^{m-k_{n}+1}K_{b_{m}}(t_{s-1}-\frac{\tau}{n})\left(\bar{Y}_{i,s}^2 - \frac{1}{2}\hat{Y}_{i,s}\right)\mathbf{1}_{\{|\bar{Y}_{i,s}|\leq \nu_{m}\}},
\end{equation}
where $K_{b}(x) = K(x/b)/b$, the bandwidth size $b_{m} = 1/n$, the weight function $g(x) = 2x \wedge (1-x)$, 
\begin{align*}
    &\bar{Y}_{i,s} = \sum_{l=1}^{k_{m}-1}g\left(\frac{l}{k_{m}}\right)(Y_{i,t_{s+l}}-Y_{i,t_{s+l-1}}), \qquad  \phi_{k_{m}}(g) = \sum_{i=1}^{k_{m}}g\left(\frac{i}{k_{m}}\right)^2,\\
    &\hat{Y}_{i,s} = \sum_{l=1}^{k_{m}}\biggl\{g\left(\frac{l}{k_{m}}\right)-g\left(\frac{l-1}{k_{m}}\right)\biggl\}^2 (Y_{i,t_{s+l}}-Y_{i,t_{s+l-1}})^2,
\end{align*}
$\mathbf{1}_{\{\cdot\}}$ is an indicator function, and $\nu_{m} = 1.8\sqrt{\text{BPV}}(k_{m}/m)^{0.47}$, where the bipower variation $\text{BPV} = \frac{\pi}{2}\sum_{s=2}^{m}|Y_{i,t_{s-1}} - Y_{i,t_{s-2}}| |Y_{i,t_{s}} - Y_{i,t_{s-1}}|$.
We used the uniform kernel function and the data-driven approach to obtain the preaveraging window size, $k_m$, as suggested in Section 3.1 of \cite{figueroa2022kernel}.

With the instantaneous volatility estimates available up to the $n_1$-th intraday data point on the $D$th day, we examined the out-of-sample performance of estimating the remaining $(1\times n_2)$ instantaneous volatility vector, where $n_1 = \omega n$ and $n_2 = n- n_1$.
For comparison, we considered the SIP, AVE, AR, SARIMA, HAR-D, XGBoost, PC, and TIP-PCA methods in predicting $c_{D,\tau}$ for $\tau = n_1+1, \ldots, n$.
AVE represents estimates obtained from the column mean of $\hat{\bSigma}_{D-1,n}$, while AR generates predictions using an autoregressive model of order 1, applied to each column of $\hat{\bSigma}_{D-1,n}$.
SARIMA provides  $n_2$-step-ahead forecasts using a seasonal ARIMA(1,1,1) model \citep{sheppard2010financial}, where the seasonal cycle length is set to $n$. 
HAR-D extends the HAR model by incorporating diurnal effects and prior intraday components in addition to daily, weekly, and monthly realized volatilities \citep{zhang2024volatility}. 
XGBoost \citep{chen2016xgboost}, a decision-tree-based ensemble learning method, accounts for nonlinear dependencies, using the same hyperparameters as in \cite{zhang2024volatility}. 
Since SARIMA, HAR-D, and XGBoost require a time series format, we first vectorized the estimates as $(\hat{c}_{1,1},\dots, \hat{c}_{D,n_1})$ before applying these methods. 
PC corresponds to the last row of the estimated low-rank matrix obtained from the best rank-$r$  approximation of $\hat{\bSigma}_{D-1,n}$.
TIP-PCA uses ex-post daily, weekly, and monthly realized volatilities and the intraday time sequence as observable covariates that explain the left and right singular vectors, as proposed by \cite{TIP-PCA}, given $\hat{\bSigma}_{D-1,n}$.
For SIP, TIP-PCA, and PC, we set rank 1 as determined by the eigenvalue ratio method \citep{ahn2013eigenvalue}.
We generated high-frequency data with  $m = 23,400$  over 200 consecutive days, using subsampled log prices from the last  $D =$ 50, 100, 150,  and  200  days. 
To assess prediction accuracy, we computed the mean squared prediction error (MSPE) as
$$
\frac{1}{n}\sum_{\tau=n_1+1}^{n}(\tilde{c}_{D,\tau}-c_{D,\tau})^{2},
$$
where  $\tilde{c}_{D,\tau}$  denotes the instantaneous volatility estimator of each method above.
Finally, we calculated the sample averages of MSPEs across 500 simulations.

\begin{figure}
	\includegraphics[width=\linewidth]{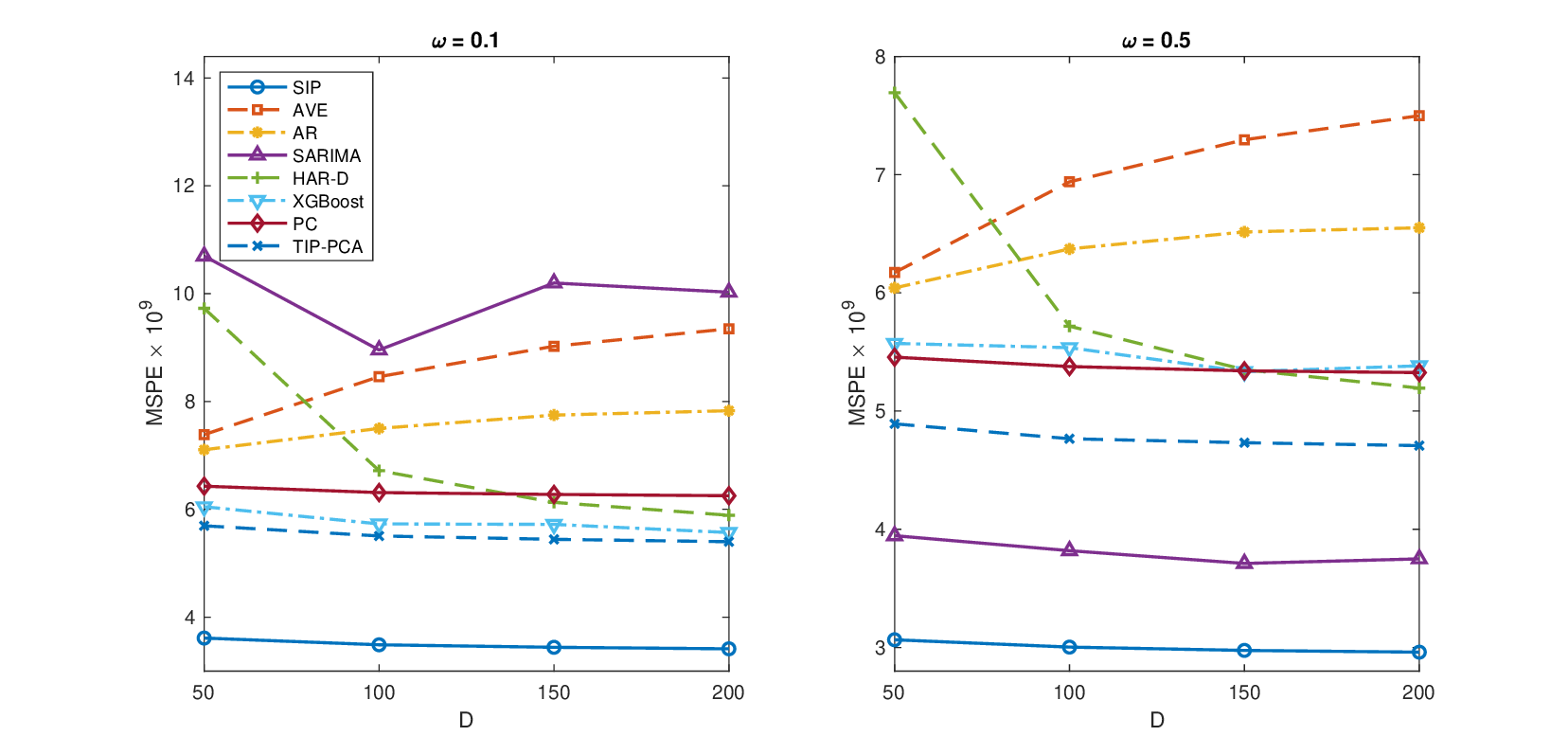}
	\centering	
	\caption{MSPE$\times 10^9$ for the SIP, AVE, AR, SARIMA, HAR-D, XGBoost, PC, and TIP-PCA against $D$ with fixed $\omega = \{0.1, 0.5\}$.}				\label{varying_D}
\end{figure}

Figure \ref{varying_D} presents the average MSPEs of the estimators for the remaining future intraday instantaneous volatility process across different values of dimensionality, $D = {50, 100, 150, 200}$, with fixed $n$ and $\omega$.
Figure \ref{varying_rho}, on the other hand, plots the average MSPEs against varying values of $\omega = {0.05, 0.1, 0.2, \dots, 0.9}$, while keeping $D$ and $n$ fixed.
Note that the target future volatility remains consistent across different $D$ values, as the same subsampled data are used in each simulation setting.
In contrast, in Figure \ref{varying_rho}, the MSPEs exhibit a U-shaped pattern as $\omega$ varies.
This behavior arises because the target future volatility differs with each value of $\omega$, and the underlying data-generating process follows a U-shaped intraday volatility pattern.
As a result, direct comparisons of MSPE values across different $\omega$ values are not appropriate.
Instead, MSPEs should be interpreted relatively across competing methods for each fixed value of $\omega$.

\begin{figure}
	\includegraphics[width=\linewidth]{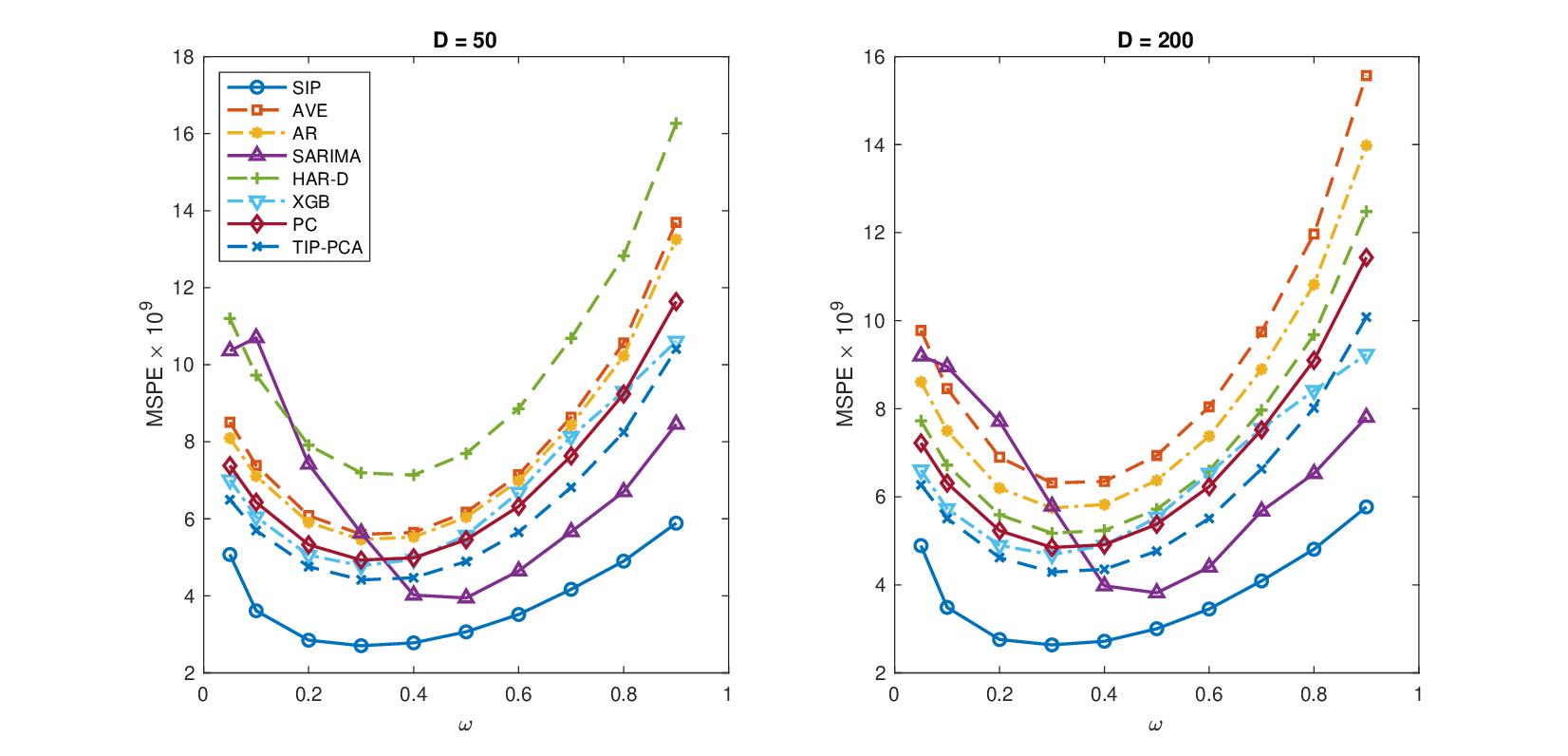}
	\centering	
	\caption{MSPE$\times 10^9$ for the SIP, AVE, AR, SARIMA, HAR-D, XGBoost, PC, and TIP-PCA against $\omega$ with fixed $D=\{50, 100\}$.}				\label{varying_rho}
\end{figure}

Figures \ref{varying_D} and \ref{varying_rho} demonstrate that the SIP method consistently achieves the best performance among the competing approaches.
This may be because by incorporating information from previous intraday volatility patterns on the $D$th day, SIP effectively predicts the remaining future instantaneous volatility process through a nonparametric framework.
Moreover, the MSPEs of SIP tend to decrease as the number of daily observations $D$ increases, which supports the theoretical results established in Section \ref{asymp}.
 In contrast, the TIP-PCA, PC, AVE, and AR methods perform poorly, as they are unable to leverage current intraday observations on the $D$th day within their estimation procedures.
In Figure \ref{varying_rho}, we observe that the performance of SARIMA improves as $\omega$ increases.
This suggests that, given sufficient real-time intraday information, SARIMA can better incorporate both periodic intraday patterns and daily autoregressive dynamics.
Nevertheless, SIP outperforms SARIMA and all other competing methods across the entire range of $\omega$, particularly when $\omega$ is small.

\section{Empirical Study} \label{empiric}
In this section, we applied the proposed SIP method to predict intraday instantaneous volatility vector using real high-frequency trading data.
Specifically, we obtained intraday data of the S\&P 500 index ETF (SPY) and the 11 Global Industry Classification Standard (GICS) sector ETFs (XLC, XLY, XLP, XLE, XLF, XLV, XLI, XLB, XLRE, XLK, and XLU) from July 2021 to June 2022, sourced from the Trade and Quote (TAQ) database in the Wharton Research Data Services (WRDS) system.
We used high-frequency data subsampled at 1-second intervals, excluding trading days with early market closures.
This subsampling mitigates the effects of irregular observation times \citep{li2023robust}.
Using the log prices, we implemented the jump robust pre-averaging estimation procedure, as defined in Section \ref{simulation}, to estimate the instantaneous variance at 5-minute intervals.
Using the in-sample period data, we applied SIP, AVE, AR, SARIMA, HAR-D, XGBoost, PC, and TIP-PCA, as described in Section \ref{simulation}, to predict the remaining instantaneous volatilities given each $\omega = \{0.1, 0.5, 0.9\}$. 
We employed the rolling window scheme with a 63-day (i.e., one quarter) in-sample period, where the last $n_2 = (1-\omega)n$ instantaneous volatilities on the final in-sample day are the target volatility.
The out-of-sample period was from October 2021 to June 2022 (i.e., $q=189$ days).

To evaluate the performance of the predicted instantaneous volatility, we employed the mean squared prediction errors (MSPE) and QLIKE loss function \citep{patton2011volatility}, defined as follows: for each $\omega$ and $n_2 = (1-\omega)n$,
\begin{align*}
    &\text{MSPE} = \frac{1}{qn_2}\sum_{i=1}^{q}\sum_{j = 1}^{n_2}(\tilde{c}_{i,j} - \hat{c}_{i,j})^{2},\\
    &\text{QLIKE} = \frac{1}{qn_2}\sum_{i=1}^{q}\sum_{j = 1}^{n_2}(\log{\tilde{c}_{i,j}} + \frac{\hat{c}_{i,j}}{\tilde{c}_{i,j}}),
\end{align*}
where $\tilde{c}_{i,j}$ represents the predicted instantaneous volatility on the $j$th intraday time point of the $i$th out-of-sample day, obtained  from one of the SIP, AVE, AR, SARIMA, HAR-D, XGBoost, PC, and TIP-PCA methods.
We predicted the remaining future conditional expected instantaneous volatilities using in-sample period data. 
Since the true conditional expected instantaneous volatility is unknown, we assessed the significance of differences in prediction performances using the Diebold and Mariano (DM) test \citep{diebold2002comparing} based on both MSPE and QLIKE. 
We compared the proposed SIP method against alternative approaches. 
Given that multiple hypothesis tests were conducted repeatedly throughout this section, it is crucial to control the False Discovery Rate (FDR). 
To account for this, we applied the Benjamini–Hochberg (BH) procedure \citep{benjamini1995controlling} at a significance level of  $\alpha = 0.05$  to adjust the  $p$-values and mitigate the risk of false positives across all hypothesis tests performed in this section.

\begin{table}[htbp]
    \centering
    \caption{MSPEs ($\times 10^9$) for SIP, AVE, AR, SARIMA, HAR-D, XGBoost, PC, and TIP-PCA across 12 ETFs under varying values of $\omega \in \{0.1, 0.5, 0.9\}$.} \label{MSPE}
        \scalebox{0.91}{
    \begin{tabular}{lllllllll}
        \hline
         & SIP & AVE & AR & SARIMA & HAR-D & XGBoost & PC & TIP-PCA \\
        \midrule
        \multicolumn{9}{c}{$\omega = 0.1$} \\
        \cmidrule(lr){2-9}
        SPY & \textbf{2.259} & 3.555\textsuperscript{***} & 3.647\textsuperscript{***} & 3.099\textsuperscript{***} & 3.590\textsuperscript{***} & 2.858\textsuperscript{***} & 2.839\textsuperscript{***} & 2.693\textsuperscript{***} \\
        XLC & \textbf{2.981} & 4.669\textsuperscript{***} & 5.073\textsuperscript{***} & 4.757\textsuperscript{***} & 5.299\textsuperscript{***} & 4.478\textsuperscript{***} & 3.873\textsuperscript{***} & 3.681\textsuperscript{***} \\
        XLY & \textbf{6.959} & 12.196\textsuperscript{***} & 11.554\textsuperscript{***} & 11.769\textsuperscript{***} & 13.272\textsuperscript{***} & 11.328\textsuperscript{***} & 9.838\textsuperscript{***} & 9.472\textsuperscript{***} \\
        XLP & \textbf{0.569} & 0.833\textsuperscript{***} & 0.816\textsuperscript{***} & 0.793\textsuperscript{***} & 0.784\textsuperscript{***} & 0.818\textsuperscript{***} & 0.641\textsuperscript{***} & 0.634\textsuperscript{***} \\
        XLE & \textbf{7.089} & 10.871\textsuperscript{***} & 10.362\textsuperscript{***} & 12.145\textsuperscript{***} & 10.327\textsuperscript{***} & 8.715\textsuperscript{***} & 8.418\textsuperscript{***} & 8.282\textsuperscript{***} \\
        XLF & \textbf{1.772} & 3.534\textsuperscript{***} & 3.079\textsuperscript{***} & 2.661\textsuperscript{***} & 3.101\textsuperscript{***} & 2.678\textsuperscript{***} & 2.595\textsuperscript{***} & 2.500\textsuperscript{***} \\
        XLV & \textbf{0.884} & 1.316\textsuperscript{***} & 1.197\textsuperscript{***} & 1.071\textsuperscript{***} & 1.251\textsuperscript{***} & 1.226\textsuperscript{***} & 0.994\textsuperscript{***} & 0.995\textsuperscript{***} \\
        XLI & \textbf{1.642} & 2.626\textsuperscript{***} & 2.626\textsuperscript{***} & 2.204\textsuperscript{***} & 2.511\textsuperscript{***} & 2.358\textsuperscript{***} & 1.900\textsuperscript{***} & 1.911\textsuperscript{***} \\
        XLB & \textbf{2.024} & 3.114\textsuperscript{***} & 3.115\textsuperscript{***} & 2.836\textsuperscript{***} & 3.005\textsuperscript{***} & 2.778\textsuperscript{***} & 2.434\textsuperscript{***} & 2.393\textsuperscript{***} \\
        XLRE & \textbf{1.558} & 2.422\textsuperscript{***} & 2.343\textsuperscript{***} & 2.087\textsuperscript{***} & 2.325\textsuperscript{***} & 2.188\textsuperscript{***} & 1.945\textsuperscript{***} & 1.918\textsuperscript{***} \\
        XLK & \textbf{6.256} & 10.667\textsuperscript{***} & 10.154\textsuperscript{***} & 8.183\textsuperscript{***} & 10.816\textsuperscript{***} & 10.222\textsuperscript{***} & 8.273\textsuperscript{***} & 8.072\textsuperscript{***} \\
        XLU & \textbf{0.814} & 1.175\textsuperscript{***} & 1.139\textsuperscript{***} & 1.108\textsuperscript{***} & 1.160\textsuperscript{***} & 1.223\textsuperscript{***} & 0.988\textsuperscript{***} & 0.981\textsuperscript{***} \\
        \midrule
        \multicolumn{9}{c}{$\omega = 0.5$} \\
        \cmidrule(lr){2-9}
        SPY & \textbf{3.012} & 4.114\textsuperscript{***} & 4.488\textsuperscript{***} & 3.125\textsuperscript{***} & 4.212\textsuperscript{***} & 3.446\textsuperscript{***} & 3.345\textsuperscript{***} & 3.234\textsuperscript{***} \\
        XLC & \textbf{3.675} & 4.857\textsuperscript{***} & 5.831\textsuperscript{***} & 3.907\textsuperscript{***} & 5.868\textsuperscript{***} & 4.668\textsuperscript{***} & 4.223\textsuperscript{***} & 4.010\textsuperscript{***} \\
        XLY & \textbf{8.327} & 12.452\textsuperscript{***} & 12.470\textsuperscript{***} & 9.307\textsuperscript{***} & 13.671\textsuperscript{***} & 10.188\textsuperscript{***} & 10.316\textsuperscript{***} & 10.136\textsuperscript{***} \\
        XLP & \textbf{0.703} & 1.005\textsuperscript{***} & 1.041\textsuperscript{***} & 0.742\textsuperscript{***} & 0.945\textsuperscript{***} & 0.972\textsuperscript{***} & 0.771\textsuperscript{***} & 0.768\textsuperscript{***} \\
        XLE & \textbf{6.196} & 8.504\textsuperscript{***} & 8.301\textsuperscript{***} & 6.747\textsuperscript{***} & 8.793\textsuperscript{***} & 6.685\textsuperscript{***} & 7.097\textsuperscript{***} & 7.035\textsuperscript{***} \\
        XLF & \textbf{2.214} & 3.892\textsuperscript{***} & 3.587\textsuperscript{***} & 2.627\textsuperscript{***} & 3.545\textsuperscript{***} & 2.861\textsuperscript{***} & 2.902\textsuperscript{***} & 2.834\textsuperscript{***} \\
        XLV & \textbf{1.022} & 1.534\textsuperscript{***} & 1.439\textsuperscript{***} & 1.044 & 1.459\textsuperscript{***} & 1.757\textsuperscript{***} & 1.179\textsuperscript{***} & 1.177\textsuperscript{***} \\
        XLI & \textbf{1.911} & 2.904\textsuperscript{***} & 3.171\textsuperscript{***} & 2.013\textsuperscript{***} & 2.750\textsuperscript{***} & 2.457\textsuperscript{***} & 2.106\textsuperscript{***} & 2.129\textsuperscript{***} \\
        XLB & \textbf{2.442} & 3.421\textsuperscript{***} & 3.611\textsuperscript{***} & 2.601\textsuperscript{***} & 3.245\textsuperscript{***} & 3.029\textsuperscript{***} & 2.667\textsuperscript{***} & 2.656\textsuperscript{***} \\
        XLRE & \textbf{1.850} & 2.765\textsuperscript{***} & 2.777\textsuperscript{***} & 2.070\textsuperscript{***} & 2.573\textsuperscript{***} & 2.536\textsuperscript{***} & 2.218\textsuperscript{***} & 2.190\textsuperscript{***} \\
        XLK & \textbf{7.727} & 11.448\textsuperscript{***} & 11.655\textsuperscript{***} & 8.401\textsuperscript{***} & 12.085\textsuperscript{***} & 9.823\textsuperscript{***} & 9.191\textsuperscript{***} & 9.067\textsuperscript{***} \\
        XLU & \textbf{0.937} & 1.344\textsuperscript{***} & 1.334\textsuperscript{***} & 0.968\textsuperscript{**} & 1.335\textsuperscript{***} & 1.250\textsuperscript{***} & 1.113\textsuperscript{***} & 1.126\textsuperscript{***} \\
    \midrule
     \multicolumn{9}{c}{$\omega = 0.9$} \\
        \cmidrule(lr){2-9}        
        SPY & \textbf{2.302} & 4.383\textsuperscript{***} & 3.623\textsuperscript{***} & 2.715 & 4.240\textsuperscript{***} & 2.555 & 3.295\textsuperscript{***} & 3.387\textsuperscript{***} \\
        XLC & \textbf{2.703} & 4.956\textsuperscript{***} & 4.519\textsuperscript{***} & 3.286\textsuperscript{*} & 5.361\textsuperscript{***} & 2.704 & 3.892\textsuperscript{***} & 3.955\textsuperscript{***} \\
        XLY & \textbf{5.483} & 11.780\textsuperscript{***} & 9.936\textsuperscript{***} & 7.261\textsuperscript{**} & 13.361\textsuperscript{***} & 5.768 & 9.009\textsuperscript{***} & 9.712\textsuperscript{***} \\
        XLP & \textbf{0.592} & 1.228\textsuperscript{***} & 1.016\textsuperscript{***} & 0.763\textsuperscript{**} & 1.116\textsuperscript{***} & 0.749\textsuperscript{**} & 0.852\textsuperscript{***} & 0.910\textsuperscript{***} \\
        XLE & 4.317 & 6.357\textsuperscript{***} & 5.428\textsuperscript{***} & \textbf{3.936} & 6.476\textsuperscript{***} & 4.050 & 5.601\textsuperscript{***} & 6.117\textsuperscript{***} \\
        XLF & \textbf{1.548} & 4.240\textsuperscript{***} & 3.015\textsuperscript{***} & 1.908\textsuperscript{**} & 3.312\textsuperscript{***} & 1.690 & 2.822\textsuperscript{***} & 2.892\textsuperscript{***} \\
        XLV & \textbf{0.943} & 1.924\textsuperscript{***} & 1.555\textsuperscript{***} & 1.102\textsuperscript{*} & 1.722\textsuperscript{***} & 1.271\textsuperscript{***} & 1.356\textsuperscript{***} & 1.445\textsuperscript{***} \\
        XLI & \textbf{1.491} & 3.456\textsuperscript{***} & 2.417\textsuperscript{***} & 1.679\textsuperscript{*} & 2.889\textsuperscript{***} & 1.984\textsuperscript{***} & 2.258\textsuperscript{***} & 2.478\textsuperscript{***} \\
        XLB & \textbf{2.011} & 4.099\textsuperscript{***} & 3.362\textsuperscript{***} & 2.405\textsuperscript{**} & 3.710\textsuperscript{***} & 2.172 & 2.931\textsuperscript{***} & 3.188\textsuperscript{***} \\
        XLRE & \textbf{1.722} & 3.692\textsuperscript{***} & 3.138\textsuperscript{***} & 1.959 & 3.373\textsuperscript{***} & 2.652\textsuperscript{***} & 2.860\textsuperscript{***} & 2.925\textsuperscript{***} \\
        XLK & \textbf{5.594} & 10.868\textsuperscript{***} & 9.580\textsuperscript{***} & 7.009\textsuperscript{*} & 11.661\textsuperscript{***} & 7.044 & 8.139\textsuperscript{***} & 8.830\textsuperscript{***} \\
        XLU & \textbf{0.916} & 1.803\textsuperscript{***} & 1.481\textsuperscript{***} & 1.211\textsuperscript{**} & 1.748\textsuperscript{***} & 1.182\textsuperscript{**} & 1.398\textsuperscript{***} & 1.523\textsuperscript{***} \\
                \bottomrule
    \end{tabular}}
       	\centering
	\begin{tablenotes}
		\item Note: Bold numbers indicate the lowest MSPE for each ETF, while $^{***}$, $^{**}$, and $^{*}$ indicate rejection of the null hypothesis against SIP at the 1\%, 5\%, and 10\% significance levels, respectively, based on the Diebold-Mariano (DM) test.
	\end{tablenotes}    
\end{table}

\begin{table}[htbp]
    \centering
    \caption{QLIKEs for SIP, AVE, AR, SARIMA, HAR-D, XGBoost, PC, and TIP-PCA across 12 ETFs under varying values of $\omega \in \{0.1, 0.5, 0.9\}$.} \label{QLIKE}
        \scalebox{0.91}{
    \begin{tabular}{lllllllll}
        \hline
         & SIP & AVE & AR & SARIMA & HAR-D & XGBoost & PC & TIP-PCA \\
        \midrule
        \multicolumn{9}{c}{$\omega = 0.1$} \\
        \cmidrule(lr){2-9}
        SPY & \textbf{-9.287} & -8.970\textsuperscript{***} & -9.114\textsuperscript{***} & -5.377\textsuperscript{***} & -7.974\textsuperscript{***} & -9.208\textsuperscript{***} & -9.213\textsuperscript{***} & -9.071\textsuperscript{***} \\
        XLC & \textbf{-9.116} & -8.863\textsuperscript{***} & -8.942\textsuperscript{***} & 0.942\textsuperscript{***} & -8.255\textsuperscript{***} & -8.990\textsuperscript{***} & -9.046\textsuperscript{***} & -8.955\textsuperscript{***} \\
        XLY & \textbf{-8.521} & -8.243\textsuperscript{***} & -8.347\textsuperscript{***} & 10.962\textsuperscript{***} & -6.636\textsuperscript{***} & -8.285\textsuperscript{***} & -8.452\textsuperscript{***} & -8.435\textsuperscript{***} \\
        XLP & \textbf{-9.576} & -9.436\textsuperscript{***} & -9.489\textsuperscript{***} & -5.412\textsuperscript{***} & -9.166\textsuperscript{***} & -9.531\textsuperscript{***} & -9.547\textsuperscript{***} & -9.543\textsuperscript{***} \\
        XLE & \textbf{-8.210} & -8.095\textsuperscript{***} & -8.131\textsuperscript{***} & -4.371\textsuperscript{***} & -7.621\textsuperscript{***} & -8.175\textsuperscript{***} & -8.179\textsuperscript{***} & -8.179\textsuperscript{***} \\
        XLF & \textbf{-8.490} & -8.385\textsuperscript{***} & -8.434\textsuperscript{***} & 1.489\textsuperscript{***} & -8.037\textsuperscript{***} & -8.465\textsuperscript{***} & -8.454\textsuperscript{***} & -8.455\textsuperscript{***} \\
        XLV & \textbf{-9.563} & -9.375\textsuperscript{***} & -9.452\textsuperscript{***} & -5.433\textsuperscript{***} & -8.594\textsuperscript{***} & -9.451\textsuperscript{***} & -9.539\textsuperscript{***} & -9.495\textsuperscript{***} \\
        XLI & \textbf{-9.286} & -9.057\textsuperscript{***} & -9.167\textsuperscript{***} & -6.441\textsuperscript{***} & -8.464\textsuperscript{***} & -9.201\textsuperscript{***} & -9.269\textsuperscript{***} & -8.926\textsuperscript{***} \\
        XLB & \textbf{-9.239} & -9.016\textsuperscript{***} & -9.088\textsuperscript{***} & -7.276\textsuperscript{***} & -7.454\textsuperscript{***} & -9.135\textsuperscript{***} & -9.202\textsuperscript{***} & -9.184\textsuperscript{***} \\
        XLRE & \textbf{-9.197} & -9.036\textsuperscript{***} & -9.077\textsuperscript{***} & -5.944\textsuperscript{***} & -7.678\textsuperscript{***} & -9.082\textsuperscript{***} & -9.137\textsuperscript{***} & -9.140\textsuperscript{***} \\
        XLK & \textbf{-8.661} & -8.374\textsuperscript{***} & -8.501\textsuperscript{***} & -6.775\textsuperscript{***} & -7.012\textsuperscript{***} & -8.569\textsuperscript{***} & -8.587\textsuperscript{***} & -8.563\textsuperscript{***} \\
        XLU & \textbf{-9.222} & -9.119\textsuperscript{***} & -9.144\textsuperscript{***} & -7.066\textsuperscript{***} & -8.418\textsuperscript{***} & -9.160\textsuperscript{***} & -9.178\textsuperscript{***} & -9.185\textsuperscript{***} \\
        \midrule
        \multicolumn{9}{c}{$\omega = 0.5$} \\
        \cmidrule(lr){2-9}
        SPY & \textbf{-9.337} & -9.006\textsuperscript{***} & -9.161\textsuperscript{***} & -9.123\textsuperscript{***} & -8.207\textsuperscript{***} & -9.225\textsuperscript{***} & -9.271\textsuperscript{***} & -9.083\textsuperscript{***} \\
        XLC & \textbf{-9.181} & -8.918\textsuperscript{***} & -9.001\textsuperscript{***} & -8.955\textsuperscript{***} & -8.168\textsuperscript{***} & -8.986\textsuperscript{***} & -9.111\textsuperscript{***} & -9.001\textsuperscript{***} \\
        XLY & \textbf{-8.640} & -8.340\textsuperscript{***} & -8.453\textsuperscript{***} & -8.033\textsuperscript{**} & -5.726\textsuperscript{***} & -8.543\textsuperscript{***} & -8.572\textsuperscript{***} & -8.541\textsuperscript{***} \\
        XLP & \textbf{-9.623} & -9.456\textsuperscript{***} & -9.520\textsuperscript{***} & -9.577\textsuperscript{***} & -9.082\textsuperscript{***} & -9.559\textsuperscript{***} & -9.591\textsuperscript{***} & -9.582\textsuperscript{***} \\
        XLE & \textbf{-8.388} & -8.265\textsuperscript{***} & -8.306\textsuperscript{***} & -8.257\textsuperscript{***} & -7.639\textsuperscript{***} & -8.355\textsuperscript{***} & -8.348\textsuperscript{***} & -8.346\textsuperscript{***} \\
        XLF & \textbf{-8.559} & -8.439\textsuperscript{***} & -8.493\textsuperscript{***} & -8.532\textsuperscript{***} & -8.046\textsuperscript{**} & -8.529\textsuperscript{***} & -8.517\textsuperscript{***} & -8.517\textsuperscript{***} \\
        XLV & \textbf{-9.657} & -9.417\textsuperscript{***} & -9.512\textsuperscript{***} & -9.430\textsuperscript{**} & -8.067\textsuperscript{***} & -9.524\textsuperscript{***} & -9.608\textsuperscript{***} & -9.547\textsuperscript{***} \\
        XLI & \textbf{-9.369} & -9.100\textsuperscript{***} & -9.223\textsuperscript{***} & -9.276\textsuperscript{***} & -8.140\textsuperscript{***} & -9.287\textsuperscript{***} & -9.341\textsuperscript{***} & -9.061\textsuperscript{***} \\
        XLB & \textbf{-9.316} & -9.064\textsuperscript{***} & -9.147\textsuperscript{***} & -9.177\textsuperscript{***} & -6.553\textsuperscript{***} & -9.212\textsuperscript{***} & -9.279\textsuperscript{***} & -9.257\textsuperscript{***} \\
        XLRE & \textbf{-9.232} & -9.048\textsuperscript{***} & -9.091\textsuperscript{***} & -8.969\textsuperscript{***} & -6.833\textsuperscript{***} & -9.111\textsuperscript{***} & -9.158\textsuperscript{***} & -9.161\textsuperscript{***} \\
        XLK & \textbf{-8.764} & -8.446\textsuperscript{***} & -8.590\textsuperscript{***} & -8.561\textsuperscript{***} & -6.418\textsuperscript{***} & -8.429\textsuperscript{*} & -8.688\textsuperscript{***} & -8.652\textsuperscript{***} \\
        XLU & \textbf{-9.294} & -9.170\textsuperscript{***} & -9.203\textsuperscript{***} & -9.231\textsuperscript{***} & -7.461\textsuperscript{***} & -9.237\textsuperscript{***} & -9.241\textsuperscript{***} & -9.245\textsuperscript{***} \\
                \midrule
        \multicolumn{9}{c}{$\omega = 0.9$} \\
        \cmidrule(lr){2-9}        
        SPY & \textbf{-9.065} & -8.754\textsuperscript{***} & -8.892\textsuperscript{***} & -9.063 & -8.485\textsuperscript{***} & -9.013\textsuperscript{**} & -8.947\textsuperscript{***} & -8.795\textsuperscript{***} \\
        XLC & \textbf{-8.858} & -8.606\textsuperscript{***} & -8.702\textsuperscript{***} & -8.819\textsuperscript{**} & -8.057 & -8.803\textsuperscript{***} & -8.760\textsuperscript{***} & -8.662\textsuperscript{***} \\
        XLY & -8.425 & -8.138\textsuperscript{***} & -8.261\textsuperscript{***} & -8.413 & -5.270\textsuperscript{***} & \textbf{-8.435} & -8.333\textsuperscript{***} & -8.304\textsuperscript{***} \\
        XLP & -9.250 & -9.089\textsuperscript{***} & -9.155\textsuperscript{***} & -9.253 & \textbf{-10.141} & -9.224\textsuperscript{***} & -9.199\textsuperscript{***} & -9.187\textsuperscript{***} \\
        XLE & -8.289 & -8.187\textsuperscript{***} & -8.224\textsuperscript{***} & -8.281 & -8.118\textsuperscript{***} & \textbf{-8.300}\textsuperscript{**} & -8.229\textsuperscript{***} & -8.229\textsuperscript{***} \\
        XLF & \textbf{-8.404} & -8.268\textsuperscript{***} & -8.340\textsuperscript{***} & -8.402 & -8.303 & -8.400 & -8.349\textsuperscript{***} & -8.347\textsuperscript{***} \\
        XLV & \textbf{-9.265} & -9.049\textsuperscript{***} & -9.124\textsuperscript{***} & -9.245 & -6.789\textsuperscript{*} & -9.236\textsuperscript{***} & -9.196\textsuperscript{***} & -9.146\textsuperscript{***} \\
        XLI & \textbf{-8.993} & -8.738\textsuperscript{***} & -8.869\textsuperscript{***} & -8.962 & -8.556\textsuperscript{***} & -8.972\textsuperscript{***} & -8.930\textsuperscript{***} & -8.623\textsuperscript{***} \\
        XLB & \textbf{-8.955} & -8.706\textsuperscript{***} & -8.807\textsuperscript{***} & -8.924 & -8.773 & -8.861 & -8.893\textsuperscript{***} & -8.858\textsuperscript{***} \\
        XLRE & \textbf{-8.804} & -8.621\textsuperscript{***} & -8.645\textsuperscript{***} & -8.680 & -1.301 & -8.709\textsuperscript{***} & -8.722\textsuperscript{***} & -8.702\textsuperscript{***} \\
        XLK & \textbf{-8.554} & -8.263\textsuperscript{***} & -8.391\textsuperscript{***} & -8.481\textsuperscript{*} & -8.245\textsuperscript{***} & -8.544 & -8.454\textsuperscript{***} & -8.418\textsuperscript{***} \\
        XLU & \textbf{-8.919} & -8.794\textsuperscript{***} & -8.824\textsuperscript{***} & -8.914 & -8.757\textsuperscript{***} & -8.893\textsuperscript{***} & -8.853\textsuperscript{***} & -8.847\textsuperscript{***} \\
                \bottomrule
    \end{tabular}}
       	\centering
	\begin{tablenotes} 
        \item Note: Bold numbers indicate the lowest QLIKE for each ETF, while $^{***}$, $^{**}$, and $^{*}$ indicate rejection of the null hypothesis against SIP at the 1\%, 5\%, and 10\% significance levels, respectively, based on the Diebold-Mariano (DM) test.
	\end{tablenotes}    
\end{table}

Tables \ref{MSPE} and \ref{QLIKE} report the MSPE and QLIKE results, respectively, under varying values of $\omega \in \{0.1, 0.5, 0.9\}$.
As discussed in Section \ref{simulation}, the target future volatility changes with $\omega$.
Consequently, direct comparisons of MSPE and QLIKE values across different $\omega$ values are inappropriate.
Instead, performance should be evaluated relative to other methods for each given $\omega$.
Tables \ref{MSPE}--\ref{QLIKE} show that the SIP method consistently outperforms other methods across different values of $\omega$. 
This implies that incorporating the early intraday trading data enhances the accuracy of instantaneous volatility prediction.
We note that the performance of SARIMA and XGBoost improves when $\omega = 0.9$.
This may be attributed to the availability of a larger amount of information on the last trading day, which improves their prediction performance for a smaller set of target volatilities.

\begin{table}[htbp]
    \centering
    \caption{Number of cases where the adjusted $p$-value exceeds 0.05 for SIP, AVE, AR, SARIMA, HAR-D, XGBoost, PC, and TIP-PCA across 12 ETFs at each $q_0 = \{0.01, 0.02, 0.05, 0.1, 0.2\}$, based on the LRuc, LRcc, and DQ tests under varying values of $\omega \in \{0.1, 0.5, 0.9\}$.} 
    \label{VaR}    
    \scalebox{0.91}{
    \begin{tabular}{lrrrrr|rrrrr|rrrrr}
        \toprule
        & \multicolumn{5}{c}{LRuc} & \multicolumn{5}{c}{LRcc} & \multicolumn{5}{c}{DQ} \\
        \cmidrule(lr){2-6} \cmidrule(lr){7-11} \cmidrule(lr){12-16}
        $q_{0}$ & 0.01 & 0.02 & 0.05 & 0.1 & 0.2 & 0.01 & 0.02 & 0.05 & 0.1 & 0.2 & 0.01 & 0.02 & 0.05 & 0.1 & 0.2 \\
        \midrule
        & \multicolumn{15}{c}{$\omega = 0.1$} \\
        \cmidrule(lr){2-16}
        SIP      & 12 & 12 & 12 & 12 & 12 & 11 & 12 & 12 & 12 & 12 & 7 & 6 & 7 & 7 & 9 \\
        AVE      &  2 &  1 &  2 & 11 & 12 &  2 &  1 &  5 &  7 & 12 & 1 & 1 & 1 & 2 & 5 \\
        AR       &  4 &  3 &  6 & 12 & 12 &  6 &  4 &  5 & 12 & 12 & 1 & 2 & 3 & 3 & 6 \\
        SARIMA   &  0 &  0 &  0 &  0 &  0 &  0 &  0 &  0 &  0 &  0 & 0 & 0 & 0 & 1 & 2 \\
        HAR-D    &  0 &  0 &  0 &  4 & 12 &  0 &  0 &  0 &  3 & 12 & 1 & 1 & 1 & 2 & 4 \\
        XGBoost  &  7 &  5 &  6 &  9 & 12 &  6 &  3 &  4 &  8 & 12 & 2 & 2 & 2 & 4 & 5 \\
        PC       &  7 &  7 & 12 & 12 & 12 &  9 & 11 & 10 & 12 & 12 & 4 & 3 & 3 & 4 & 7 \\
        TIP-PCA  & 12 & 10 & 12 & 12 & 12 &  9 & 10 & 11 & 12 & 12 & 4 & 3 & 4 & 5 & 7 \\
        \midrule       
        & \multicolumn{15}{c}{$\omega = 0.5$} \\
        \cmidrule(lr){2-16}
        SIP       & 12 & 12 & 12 & 12 & 12 & 11 & 12 & 12 & 12 & 12 & 6  & 7  & 7  & 9  & 10 \\
        AVE       & 3  & 3  & 5  & 12 & 12 & 5  & 3  & 8  & 11 & 12 & 2  & 2  & 3  & 4  & 6  \\
        AR        & 7  & 6  & 11 & 12 & 12 & 7  & 8  & 11 & 12 & 12 & 3  & 3  & 4  & 5  & 8  \\
        SARIMA    & 5  & 6  & 11 & 12 & 12 & 5  & 7  & 12 & 12 & 12 & 2  & 3  & 4  & 7  & 10 \\
        HAR-D     & 0  & 0  & 1  & 4  & 12 & 0  & 1  & 1  & 3  & 12 & 1  & 1  & 2  & 3  & 5  \\
        XGBoost   & 5  & 5  & 4  & 8  & 12 & 6  & 5  & 5  & 11 & 12 & 2  & 2  & 3  & 5  & 7  \\
        PC        & 10 & 12 & 12 & 12 & 12 & 10 & 12 & 12 & 12 & 12 & 5  & 5  & 5  & 6  & 8  \\
        TIP-PCA   & 10 & 12 & 12 & 12 & 12 & 10 & 10 & 11 & 12 & 12 & 5  & 4  & 5  & 6  & 8  \\
        \midrule       
        & \multicolumn{15}{c}{$\omega = 0.9$} \\
        \cmidrule(lr){2-16}
        SIP      & 12 & 12 & 12 & 12 & 12 & 12 & 12 & 12 & 12 & 12 & 8  & 9  & 10 & 11 & 12 \\
        AVE      & 12 & 12 & 12 & 12 & 12 & 11 & 11 & 12 & 12 & 12 & 6  & 6  & 7  & 9  & 10 \\
        AR       & 12 & 12 & 12 & 12 & 12 & 12 & 12 & 12 & 12 & 12 & 7  & 8  & 8  & 10 & 11 \\
        SARIMA   & 12 & 12 & 12 & 12 & 12 & 12 & 12 & 12 & 12 & 12 & 7  & 7  & 9  & 11 & 11 \\
        HAR-D    & 6  & 8  & 12 & 12 & 12 & 6  & 8  & 12 & 12 & 12 & 4  & 4  & 6  & 9  & 11 \\
        XGBoost  & 12 & 12 & 12 & 12 & 12 & 12 & 12 & 12 & 12 & 12 & 7  & 8  & 8  & 10 & 11 \\
        PC       & 12 & 12 & 12 & 12 & 12 & 12 & 12 & 12 & 12 & 12 & 6  & 7  & 8  & 10 & 11 \\
        TIP-PCA  & 12 & 12 & 12 & 12 & 12 & 11 & 12 & 12 & 12 & 12 & 7  & 7  & 9  & 10 & 11 \\
        \bottomrule
    \end{tabular}}
\end{table}

To further evaluate the performance of the proposed method, we conducted a 5-minute frequency Value at Risk (VaR) estimation for the remainder of the trading day, given each $\omega$.
Specifically, we first predicted the remaining conditional expected instantaneous volatilities on the $D$th day using the SIP, AVE, AR, SARIMA, HAR-D, XGBoost, PC, and TIP-PCA, based on the same in-sample period data as in the previous analysis.
We then computed quantiles using historical standardized 5-minute returns.
Specifically, we first standardized in-sample 5-minute returns with the estimated conditional instantaneous volatilities; then, we derived sample quantiles at levels 0.01, 0.02, 0.05, 0.1, and 0.2.
Using these sample quantiles and predicted instantaneous volatilities, we obtained 5-minute frequency VaR values for each prediction method. 
We used a fixed in-sample period as one quarter and employed a rolling window scheme with the out-of-sample period spanning 189 days, which is consistent with the previous analysis.

Based on the estimated VaR values, we conducted the likelihood ratio unconditional coverage (LRuc) test \citep{kupiec1995techniques}, the likelihood ratio conditional coverage (LRcc) test \citep{christoffersen1998evaluating}, and the dynamic quantile (DQ) test with lag 4 \citep{engle2004caviar}.
Table \ref{VaR} reports the number of cases where the adjusted $p$-value using the BH procedure exceeds 0.05 for the 12 ETFs at each $q_0 \in \{0.01, 0.02, 0.05, 0.1,0.2\}$, based on the LRuc, LRcc, and DQ tests under varying values of $\omega \in \{0.1,0.5,0.9\}$.
From Table \ref{VaR}, we find that the SIP method consistently outperforms other methods across all hypothesis tests.
We note that when $\omega$ is small, the TIP-PCA method also performs well relative to other methods.
This may be because when $\omega$ is small, the amount of current trading day's information is small. 
Thus, TIP-PCA, which uses only the information up to the previous day's information, can account for the interday and intraday dynamics by incorporating the interday HAR structure and the U-shaped intraday volatility feature.
Additionally, Table \ref{VaR} shows that the performance of SARIMA gets better as $\omega$ grows, aligning with the results in Table \ref{MSPE}.
However, SARIMA exhibits poor performance when $\omega$ is small in terms of VaR estimation.
Table \ref{VaR} indicates that SIP outperforms TIP-PCA and other methods, particularly for lower quantiles, which are generally more difficult to predict.
This result confirms that the proposed nonparametric prediction method, SIP, by leveraging current intraday volatility information, significantly enhances both the prediction accuracy of future instantaneous volatilities and the effectiveness of risk management.

\section{Conclusion} \label{conclusion}

This paper proposes a novel nonparametric prediction procedure for forecasting the remaining future instantaneous volatility process during the current intraday trading period. 
The proposed method, Structural Intraday-volatility Prediction (SIP), leverages both previous days' data and partially observed current-day data by considering the missing component structure of the approximately low-rank matrix representation of the instantaneous volatility process.
That is, this paper extends the low-rank matrix completion to predict time series patterns, which makes it possible to forecast future intraday volatility without relying on parametric modeling assumptions.
As a result, SIP is robust to model misspecification.
We establish the asymptotic properties of the SIP estimator and prove its consistency.

In the empirical analysis, the SIP method outperforms existing alternatives in terms of out-of-sample forecasting accuracy and Value at Risk (VaR) estimation for the remaining intraday volatility.
Notably, SIP demonstrates strong predictive performance even when only a small fraction of intraday observations is available.
This finding highlights that early intraday information plays a crucial role in forecasting the remaining volatility, and that the low-rank structure effectively captures intraday volatility dynamics without requiring a specific model form such as an autoregressive process.

	\bibliography{myReferences}

\newpage
\appendix

\section{Appendix} \label{proofs}

	\subsection{Related Lemmas} \label{proof of Thm1}
    We first present useful lemmas below.
    Let $\hat{A}_{\bullet 1}$ be the best rank-$r$ approximation to $\hat{\Sigma}_{\bullet 1}$ such that $\hat{A}_{\bullet 1} = \sum_{i=1}^{r}\hat{\lambda}_{1,i}\hat{u}_{1,i}\hat{v}_{1,i}^{\top}$, where $\{\hat{\lambda}_{1,i}, \hat{u}_{1,i}, \hat{v}_{1,i}\}_{i=1}^{D\wedge n_1}$ are the ordered singular values, left-singular vectors and right-singular vectors of $\hat{\Sigma}_{\bullet 1}$ in decreasing order.
    Similarly, we let $\hat{A}_{1 \bullet}$ be the best rank-$r$ approximation to $\hat{\Sigma}_{1 \bullet}$.
    Then, we can write the SIP estimator as follows: 
    $$
    \tilde{\Sigma}_{22} = \hat{\Sigma}_{21}\hat{V}_{11}(\hat{U}_{11}^{\top}\hat{\Sigma}_{11}\hat{V}_{11})^{-1}\hat{U}_{11}^{\top}\hat{\Sigma}_{12} = \hat{A}_{21}\hat{V}_{11}(\hat{U}_{11}^{\top}\hat{\Sigma}_{11}\hat{V}_{11})^{-1}\hat{U}_{11}^{\top}\hat{A}_{12},
    $$
    where $\hat{A}_{21}$ is the last row of $\hat{A}_{\bullet 1}$, and $\hat{A}_{12}$ corresponds to the $(D-1) \times n_2$ right block  of $\hat{A}_{1 \bullet}$.
    
    Let $\Lambda_{\bullet 1}$ and $\Lambda_{1 \bullet}$ denote the $r \times r$ diagonal matrices of the leading $r$ singular values of $A_{\bullet 1} \equiv [A_{11}^{\top} A_{21}^{\top}]^{\top}$ and $A_{1\bullet} \equiv [A_{11} A_{12}]$, respectively.
    Define $\hat{\Lambda}_{\bullet 1} = \Diag(\hat{\lambda}_{\bullet 1}, \dots, \hat{\lambda}_{\bullet r}),$ where $\hat{\lambda}_{\bullet 1}\geq \hat{\lambda}_{\bullet2}\geq \cdots \geq \hat{\lambda}_{\bullet r}$ are the square root of leading eigenvalues of $\hat{\Sigma}_{\bullet 1}\hat{\Sigma}_{\bullet 1}^{\top}$. 
    Similarly, define $\hat{\Lambda}_{1 \bullet} = \Diag(\hat{\lambda}_{1\bullet}, \dots, \hat{\lambda}_{r \bullet})$, where $\hat{\lambda}_{1 \bullet}\geq \hat{\lambda}_{2\bullet}\geq \cdots \geq \hat{\lambda}_{r\bullet}$ are the square root of leading eigenvalues of $\hat{\Sigma}_{1\bullet }^{\top}\hat{\Sigma}_{1\bullet}$.
    The following lemma presents the individual convergence rate of singular value estimators.    

\begin{lem}\label{sin_val}
        Under Assumption \ref{assum_spotvol}, we have
        \begin{align*}
            \text{(i) }  \|\hat{\Lambda}_{\bullet 1} - \Lambda_{\bullet 1}\|_{\max} &= O_{P}\Bigg(\sqrt{\frac{D}{n_{1}}}+\sqrt{n_{1}D}\rho_{m} + \sqrt{\varphi_{D}} + \sqrt{\frac{n_{1}}{D}}\varphi_D\Bigg),\\
            \text{(ii) }  \|\hat{\Lambda}_{1\bullet} - \Lambda_{1\bullet}\|_{\max} &= O_{P}\Bigg(\sqrt{{\frac{n}{D}}} +\sqrt{nD}\rho_{m}  + \sqrt{\varphi_{n}} + \sqrt{\frac{D}{n}}\varphi_n\Bigg).
        \end{align*}
    \end{lem}
    \begin{proof}
    We first consider (i). 
    We have
    \begin{align}
        &E\|\hat{\Sigma}_{\bullet 1}\hat{\Sigma}_{\bullet 1}^{\top} - \Sigma_{\bullet 1}\Sigma_{\bullet 1}^{\top}\|_{F}^{2} \nonumber \\
        & = \sum_{i=1}^{D}\sum_{s=1}^{D}E\left [ \left (\sum_{j=1}^{n_1}(\hat{c}_{i,j}\hat{c}_{s,j} - c_{i,j}c_{s,j}) \right)^{2} \right ] \nonumber\\
        & = \sum_{i=1}^{D}E \left [ \left  (\sum_{j=1}^{n_1}(\hat{c}_{i,j}^2 - c_{i,j}^2) \right )^{2}  \right ] + \sum_{i=1}^{D}\sum_{s\neq i}^{D}E \left [ \left (\sum_{j=1}^{n_1}(\hat{c}_{i,j}\hat{c}_{s,j} - c_{i,j}c_{s,j})\right )^{2} \right ]  \nonumber\\
        & = \sum_{i=1}^{D}\sum_{j=1}^{n_1}E \left [ (\hat{c}_{i,j}^2 - c_{i,j}^2 )^{2}  \right ] +\sum_{i=1}^{D}\sum_{j=1}^{n_1}\sum_{k\neq j }^{n_1}E \left [ (\hat{c}_{i,j}^2 - c_{i,j}^2 )(\hat{c}_{i,k}^2 - c_{i,k}^2 )  \right ] \nonumber\\
        & \qquad+ \sum_{i=1}^{D}\sum_{s\neq i}^{D}E \left [ \left (\sum_{j=1}^{n_1}(\hat{c}_{i,j}\hat{c}_{s,j} - c_{i,j}c_{s,j})\right )^{2} \right ] \cr
        & := I + II + III  .     \nonumber    
    \end{align}
     For $I$, by substituting $\hat{c}_{i,j} = c_{i,j} + \upsilon_{i,j} + \varsigma_{i,j}$ and using the Cauchy–Schwarz inequality with the sub-Gaussian conditions from Assumption \ref{assum_spotvol}(iii), we can show 
    \begin{align*}
        I & = \sum_{i=1}^{D}\sum_{j=1}^{n_1}E \left [ (\hat{c}_{i,j}^2 - c_{i,j}^2 )^{2}  \right ] = \sum_{i=1}^{D}\sum_{j=1}^{n_1}E \left [ (2c_{i,j}\upsilon_{i,j}+2c_{i,j}\varsigma_{i,j} + \upsilon_{i,j}^{2}+\varsigma_{i,j}^{2}+2\upsilon_{i,j}\varsigma_{i,j})^{2} \right ]\\
        & = O(Dn_1 m^{-1/4}+Dn_1 \rho_{m}^{2}). 
    \end{align*}
     Similarly, we can show that
    \begin{align*}
         II &=\sum_{i=1}^{D}\sum_{j=1}^{n_1}\sum_{k\neq j }^{n_1}E \left [ (\hat{c}_{i,j}^2 - c_{i,j}^2 )(\hat{c}_{i,k}^2 - c_{i,k}^2 )  \right ] \cr
         &\leq  \sum_{i=1}^{D}\sum_{j=1}^{n_1}\sum_{k\neq j }^{n_1} E \left [ (\hat{c}_{i,j}^2 - c_{i,j}^2 ) ^2 \right ] ^{1/2}  E \left [  (\hat{c}_{i,k}^2 - c_{i,k}^2 )^2  \right ]^{1/2}\cr
        &= O(  Dn_1^{2}m^{-1/4} + Dn_1^{2}\rho_{m}^2),
    \end{align*}
    where the Hölder's inequality is used for the first inequality.

    For $III$, we first consider the martingale part only (i.e., $\varsigma_{i,j}$ is zero).
    Then, we can extend the result to the general case.
    Without loss of generality, we assume that $i>s$. 
    We have
    \begin{align*}
        &  \hat{c}_{i,j}\hat{c}_{s,j} - c_{i,j}c_{s,j}  =    ( \hat{c}_{i,j}  - c_{i,j}) (\hat{c}_{s,j} -c_{s,j}) +  ( \hat{c}_{i,j}  - c_{i,j})   c_{s,j}   +   c_{i,j}  (\hat{c}_{s,j}- c_{s,j}). 
    \end{align*}
    Then, for any $j > j'$, using the martingale difference property, we can show
    \begin{align*}
    &E [ (\hat{c}_{i,j}\hat{c}_{s,j} - c_{i,j}c_{s,j})(\hat{c}_{i,j'}\hat{c}_{s,j'} - c_{i,j'}c_{s,j'})]\cr
    &=E [ ( \hat{c}_{i,j'}  - c_{i,j'}) (\hat{c}_{s,j'} -c_{s,j'})  c_{i,j}  (\hat{c}_{s,j}- c_{s,j})] + E [  ( \hat{c}_{i,j'}  - c_{i,j'})   c_{s,j'}  c_{i,j}  (\hat{c}_{s,j}- c_{s,j})] \cr
    &\quad + E [ c_{i,j'}  (\hat{c}_{s,j'}- c_{s,j'})c_{i,j}  (\hat{c}_{s,j}- c_{s,j}) ] \cr
    &= E [ E [ ( \hat{c}_{i,j'}  - c_{i,j'})   c_{i,j}   | \cF_{i,j'}] (\hat{c}_{s,j'} -c_{s,j'})  (\hat{c}_{s,j}- c_{s,j})] \cr
    &\quad + E [  E [ ( \hat{c}_{i,j'}  - c_{i,j'})   c_{i,j}   | \cF_{i,j'}]     c_{s,j'}   (\hat{c}_{s,j}- c_{s,j})]   + E [ E[ c_{i,j'}  c_{i,j}  (\hat{c}_{s,j}- c_{s,j}) | \cF_{s,j}] (\hat{c}_{s,j'}- c_{s,j'})] \cr
    &=0,
        \end{align*}
where the last equality is due to Assumption \ref{assum_spotvol}(iii),
    and we have 
    \begin{align*}
    \sum_{i=1}^{D}\sum_{s\neq i}^{D} \sum_{j=1}^{n_1} E \left [ (\hat{c}_{i,j}\hat{c}_{s,j} - c_{i,j}c_{s,j}) ^2 \right ] = O(D^2 n_1 m^{-1/4}). 
    \end{align*}
For the general case, since $\hat{c}_{i,j}\hat{c}_{s,j} - c_{i,j}c_{s,j}$ can be decomposed into the martingale difference term and the bias term, with their respective effects of $m^{-1/8}$ and $\rho_m$, we have
\begin{align*}
    III &= \sum_{i=1}^{D}\sum_{s\neq i}^{D}E \left [ \left (\sum_{j=1}^{n_1}(\hat{c}_{i,j}\hat{c}_{s,j} - c_{i,j}c_{s,j})\right )^{2} \right ]\\
    & = \sum_{i=1}^{D}\sum_{s\neq i}^{D}\sum_{j=1}^{n_1}E \left [ \left (\hat{c}_{i,j}\hat{c}_{s,j} - c_{i,j}c_{s,j}\right )^{2} \right ]\\
    & \qquad + \sum_{i=1}^{D}\sum_{s\neq i}^{D}\sum_{j=1}^{n_1}\sum_{k\neq j}^{n_1}E \left [ \left (\hat{c}_{i,j}\hat{c}_{s,j} - c_{i,j}c_{s,j}\right ) \left (\hat{c}_{i,k}\hat{c}_{s,k} - c_{i,k}c_{s,k}\right )\right ]\\
    & = O(D^{2}n_1m^{-1/4} + D^{2}n_1\rho_{m}^{2}) + O(D^{2}n_1^{2}m^{-\frac{1}{8}}\rho_{m} + D^2 n_1^2 \rho_m^2)\\
    & = O(D^2 n_1 m^{-1/4} +D^{2}n_1^{2}m^{-1/8}\rho_{m}+ D^2 n_1^2 \rho_m^2),
\end{align*}
where the cross product terms such as $E(c_{i,j}\varsigma_{i,j'}\upsilon_{s,j}c_{s,j'})$ and $E(c_{i,j}c_{i,j'}\varsigma_{s,j}\varsigma_{s,j'})$ yield the rate above.
 Therefore, we have 
    \begin{align}
        E\|\hat{\Sigma}_{\bullet 1}\hat{\Sigma}_{\bullet 1}^{\top} - \Sigma_{\bullet 1}\Sigma_{\bullet 1}^{\top}\|_{F}^{2} = O(Dn_1^{2}m^{-\frac{1}{4}}+ D^{2}n_1m^{-\frac{1}{4}}+ D^{2}n_1^{2}m^{-\frac{1}{8}}\rho_{m} + D^{2}n_1^2\rho_{m}^{2}). \label{fro_norm_sample_cov}
    \end{align}

      We define $U_{\bullet 1} = (u_{\bullet 1},\ldots, u_{\bullet r})$ and $V_{\bullet 1}= (v_{\bullet 1},\ldots, v_{\bullet r})$ as the left and right singular vector matrices of  $A_{\bullet 1} \equiv [A_{11}^{\top} A_{21}^{\top}]^{\top}$, respectively, and $E_{\bullet 1}= [E_{11}^{\top}E_{21}^{\top}]^{\top}$.
      The columns of $\hat{U}_{\bullet 1} = (\hat{u}_{\bullet 1},\ldots, \hat{u}_{\bullet r})$ are defined as the top $r$ eigenvectors of $\hat{\Sigma}_{\bullet 1}\hat{\Sigma}_{\bullet 1}^{\top}$.
       Let $\Lambda_{\bullet 1} = \Diag(\lambda_{\bullet 1}, \dots,\lambda_{\bullet r})$.
      For $k \leq r$, we have
    \begin{align*}
        |\hat{\lambda}_{\bullet k}^{2} - \lambda_{\bullet k}^2| &= |\hat{u}_{\bullet k}^{\top}\hat{\Sigma}_{\bullet 1}\hat{\Sigma}_{\bullet 1}^{\top}\hat{u}_{\bullet k} - u_{\bullet k}^{\top}U_{\bullet 1}\Lambda_{\bullet 1}^{2}U_{\bullet 1}^{\top}u_{\bullet k}|\\
        & \leq |\hat{u}_{\bullet k}^{\top}\hat{\Sigma}_{\bullet 1}\hat{\Sigma}_{\bullet 1}^{\top}\hat{u}_{\bullet k} - u_{\bullet k}^{\top}\hat{\Sigma}_{\bullet 1}\hat{\Sigma}_{\bullet 1}^{\top}u_{\bullet k}| + |u_{\bullet k}^{\top}\hat{\Sigma}_{\bullet 1}\hat{\Sigma}_{\bullet 1}^{\top}u_{\bullet k} - u_{\bullet k}^{\top}\Sigma_{\bullet 1}\Sigma_{\bullet 1}^{\top}u_{\bullet k}|\\
        & \qquad + |u_{\bullet k}^{\top}\Sigma_{\bullet 1}\Sigma_{\bullet 1}^{\top}u_{\bullet k} - u_{\bullet k}^{\top}U_{\bullet 1}\Lambda_{\bullet 1}^{2}U_{\bullet 1}^{\top}u_{\bullet k}|\\
        & := I + II + III.        
    \end{align*}
    We note that,  by using the $\sin(\theta)$ theorem of \cite{davis1970rotation}, for $k\leq r$, we have 
    \begin{align}
        \|\hat{u}_{\bullet k}-u_{\bullet k}\|_{2} &= O\left(\frac{\|\hat{\Sigma}_{\bullet 1}\hat{\Sigma}_{\bullet 1}^{\top} - U_{\bullet 1}\Lambda_{\bullet 1}^{2}U_{\bullet 1}^{\top}\|}{n_{1}D}\right)\nonumber\\
        & = O\left(\frac{\|\hat{\Sigma}_{\bullet 1}\hat{\Sigma}_{\bullet 1}^{\top} - \Sigma_{\bullet 1}\Sigma_{\bullet 1}^{\top}\|_{F} + \| \Sigma_{\bullet 1}\Sigma_{\bullet 1}^{\top}-\sum_{j=1}^{n_1}\bOmega_{j,D\times D}\| +  \|\sum_{j=1}^{n_1}\bOmega_{e_j,D\times D}\|}{n_{1}D}\right)\nonumber\\
        & = O_{P}\left(\sqrt{\frac{1}{Dm^{\frac{1}{4}}} + \frac{1}{n_{1}m^{\frac{1}{4}}} + m^{-\frac{1}{8}}\rho_{m} + \rho_{m}^{2}}\right)+ O_{P}\left(\frac{1}{\sqrt{n_{1}}} + \frac{1}{\sqrt{D}}\right)+ O_{P}\left(\frac{\varphi_{D}}{D}\right)\nonumber\\
        & = O_{P}\left(\sqrt{\frac{\rho_{m}}{m^{\frac{1}{8}}}} + \rho_{m} + \frac{1}{\sqrt{n_{1}}} + \frac{1}{\sqrt{D}}+ \frac{\varphi_{D}}{D}\right),\label{sin theorem}
    \end{align}
    where the third equality follows from \eqref{fro_norm_sample_cov}, Assumption \ref{assum_spotvol}(v), and the following bound:
    \begin{align*}
    &\|\Sigma_{\bullet 1}\Sigma_{\bullet 1}^{\top}-\sum_{j=1}^{n_1}\bOmega_{j,D\times D}\| \leq 2\|U_{\bullet 1}\Lambda_{\bullet 1}V_{\bullet 1}^{\top}E_{\bullet 1}^{\top}\| + \|E_{\bullet 1}E_{\bullet 1}^{\top} - \sum_{j=1}^{n_{1}}\bOmega_{e_{j},D\times D}\| \\
    &\qquad\qquad\qquad = O_{P}(\sqrt{Dn_{1}}\max(\sqrt{D},\sqrt{n_{1}})) + O_{P}(\varphi_{D}\sqrt{Dn_{1}}) = O_{P}(D\sqrt{n_{1}}+n_{1}\sqrt{D}).     
    \end{align*}
    For $I$, we can show
    \begin{align*}
        I & \leq |(\hat{u}_{\bullet k} - u_{\bullet k})^{\top}\hat{\Sigma}_{\bullet 1}\hat{\Sigma}_{\bullet 1}^{\top}(\hat{u}_{\bullet k} - u_{\bullet k})| + 2|\hat{u}_{\bullet k}^{\top}\hat{\Sigma}_{\bullet 1}\hat{\Sigma}_{\bullet 1}^{\top}u_{\bullet k} - \hat{u}_{\bullet k}^{\top}\hat{\Sigma}_{\bullet 1}\hat{\Sigma}_{\bullet 1}^{\top}\hat{u}_{\bullet k}|\\
        & = |(\hat{u}_{\bullet k} - u_{\bullet k})^{\top}\hat{\Sigma}_{\bullet 1}\hat{\Sigma}_{\bullet 1}^{\top}(\hat{u}_{\bullet k} - u_{\bullet k})| + 2|\hat{\lambda}_{\bullet k}^2 \hat{u}_{\bullet k}^{\top}u_{\bullet k} - \hat{\lambda}_{\bullet k}^2|\\
        & \leq \hat{\lambda}_{\bullet k}^2\|\hat{u}_{\bullet k}-u_{\bullet k}\|_{2}^{2} + 2\hat{\lambda}_{\bullet k}^{2}|1-\hat{u}_{\bullet k}^{\top}u_{\bullet k}|\\
        & \leq 2\hat{\lambda}_{\bullet k}^2\|\hat{u}_{\bullet k}-u_{\bullet k}\|_{2}^{2}\\
        & = O_{P}\left( \frac{Dn_{1}\rho_{m}}{m^{\frac{1}{8}}} + Dn_{1}\rho_{m}^2 + D + n_{1} + \frac{n_{1}\varphi_{D}^{2}}{D}\right),
    \end{align*}
    where the last equality is due to \eqref{sin theorem} above.

    We define $\Upsilon = (\upsilon_{i,j})_{D\times n_{1}}$ and $\Theta = (\varsigma_{i,j})_{D\times n_{1}}$. 
    Let $u_{\bullet,ik}$ and $v_{\bullet, jk}$ denote the $(i,k)$ entry of $U_{\bullet 1}$ and the $(j,k)$ entry of $V_{\bullet 1}$, respectively.
    By Assumption \ref{assum_spotvol}(ii)--(iii), we have 
    $$
    E(u_{\bullet k}^\top \Upsilon v_{\bullet k}) = E(\sum_{i=1}^D \sum_{j=1}^{n_{1}} u_{\bullet, ik} v_{\bullet, jk} \upsilon_{i,j}) =  0,
    $$ 
    and for $j<j'$, 
    $$
    E(\upsilon_{i,j} \upsilon_{i,j'}) = E(\upsilon_{i,j}E(\upsilon_{i,j'}|\cF_{i,t_{j'}})) = 0.
    $$
    Therefore, we have
    \begin{align*}
    E[(\sum_{i=1}^D \sum_{j=1}^{n_{1}} u_{\bullet, ik} v_{\bullet, jk} \upsilon_{i,j})^2]&= \sum_{i,i'=1}^D \sum_{j,j'=1}^{n_1} u_{\bullet, ik} u_{\bullet, i'k} v_{\bullet, jk} v_{\bullet, j'k} E[\upsilon_{i,j} \upsilon_{i',j'}] \cr    
    &= \sum_{i=1}^D \sum_{j=1}^{n_{1}} u_{\bullet ik}^2 v_{\bullet, jk}^2 E[\upsilon_{i,j}^2] = O(m^{-1/4}).
    \end{align*}
    Hence, by Chebyshev’s inequality, we obtain 
    $$
    |u_{\bullet k}^\top \Upsilon v_{\bullet k}| = O_{P}(m^{-1/8}).
    $$
    In addition, we have  
    \begin{eqnarray*}
    E(\|u_{\bullet k}^{\top}\Upsilon\|_{2}^{2}) = \sum_{j=1}^{n_{1}}E[(\sum_{i=1}^{D}u_{\bullet, ik}\upsilon_{i,j})^{2}]=\sum_{j=1}^{n_{1}}\sum_{i=1}^{D}u_{\bullet, ik}^2E(\upsilon_{i,j}^{2}) = O(n_{1}m^{-1/4})
    \end{eqnarray*}
    and 
    $$
    E(\|E_{\bullet 1}^{\top}u_{\bullet k}\|_{2}^2) = \sum_{j=1}^{n_{1}}E[(\be_{\cdot j}^{\top}u_{\bullet k})^2] = \sum_{j=1}^{n_{1}}u_{\bullet k}^{\top}\bOmega_{e_j,D\times D}u_{\bullet k} \leq n_{1}\|\bOmega_{e_j,D \times D}\|\|u_{\bullet k}\|^2 = O(n_{1}\varphi_{D}).
    $$
    Moreover, 
    $$
    |u_{\bullet k}^\top \Theta v_{\bullet k}| \leq \sum_{i=1}^{D}\sum_{j=1}^{n_{1}}|u_{\bullet, ik}||v_{\bullet, jk}||\varsigma_{i,j}| \leq \frac{C}{\sqrt{n_{1}D}}\sum_{i=1}^{D}\sum_{j=1}^{n_{1}}|\varsigma_{i,j}| = O_{P}(\sqrt{n_{1}D}\rho_{m}).
    $$
    Since $E(\sum_{i=1}^{D}u_{\bullet, ik}\varsigma_{i,j})\leq E(\sum_{i=1}^{D}|u_{\bullet, jk}||\varsigma_{i,j}|) = O(\sqrt{D}\rho_{m})$, we obtain, by Cauchy-Schwarz inequality, 
    $$
    E(\|u_{\bullet k}^{\top}\Theta\|_{2}^{2}) = \sum_{j=1}^{n_{1}}E[(\sum_{i=1}^{D}u_{\bullet, ik}\varsigma_{i,j})^{2}] \leq \sum_{j=1}^{n_{1}} (E(\sum_{i=1}^{D}u_{\bullet, ik}\varsigma_{i,j}))^{2} = O(n_{1}D\rho_{m}^{2}).
    $$
    Using the facts above, we have
    \begin{align*}
        | u_{\bullet k}^\top \Upsilon \Sigma_{\bullet 1}^\top u_{\bullet k} |  &\leq | u_{\bullet k}^\top \Upsilon V\Lambda U^{\top} u_{\bullet k} | + | u_{\bullet k}^\top \Upsilon E_{\bullet 1}^\top u_{\bullet k} | = | u_{\bullet k}^\top \Upsilon v_{\bullet k}\lambda_{\bullet k} | + | u_{\bullet k}^\top \Upsilon E_{\bullet 1}^\top u_{\bullet k} | \\
        &\leq | u_{\bullet k}^\top \Upsilon v_{\bullet k}|\lambda_{\bullet k}  + \|u_{\bullet k}^\top \Upsilon\|_{2}\|E_{\bullet 1}^\top u_{\bullet k} \|_{2} = O_{P}(\sqrt{n_{1}D}m^{-1/8} + n_{1}m^{-1/8}\sqrt{\varphi_{D}}),\\
        | u_{\bullet k}^\top \Theta \Sigma^\top u_{\bullet k} |  &\leq | u_{\bullet k}^\top \Theta V\Lambda U^{\top} u_{\bullet k} | + | u_{\bullet k}^\top \Theta E_{\bullet 1}^\top u_{\bullet k} |  \leq | u_{\bullet k}^\top \Theta v_{\bullet k}|\lambda_{\bullet k}  + \|u_{\bullet k}^\top \Theta\|_{2}\|E_{\bullet 1}^\top u_{\bullet k} \|_{2} \\
        & = O_{P}(n_{1}D\rho_{m} + n_{1}\rho_{m}\sqrt{D\varphi_{D}})= O_{P}(n_{1}D\rho_{m}),\\
        | u_{\bullet k}^\top \Upsilon \Theta^\top u_{\bullet k} |  & \leq \|u_{\bullet k}^{\top}\Upsilon\|_{2}\|u_{\bullet k}^{\top}\Theta\|_{2} = O_{P}(n_{1}\sqrt{D}m^{-1/8}\rho_{m}),\\
        | u_{\bullet k}^\top \Upsilon \Upsilon^\top u_{\bullet k} |  & = \|u_{\bullet k}^{\top} \Upsilon\|_{2}^2 = O_{P}(n_{1}m^{-1/4}), \text{ and }
        | u_{\bullet k}^\top \Theta \Theta^\top u_{\bullet k} | = \|u_{\bullet k}^{\top} \Theta\|_{2}^2 = O_{P}(n_{1}D\rho_{m}^2).
    \end{align*}
    Therefore, by substituting $\hat{\Sigma}_{\bullet 1} = \Sigma_{\bullet 1} + \Upsilon + \Theta$, we have
    \begin{align*}
        II &= |u_{\bullet k}^{\top}(\hat{\Sigma}_{\bullet 1}\hat{\Sigma}_{\bullet 1}^{\top}-\Sigma_{\bullet 1}\Sigma_{\bullet 1}^{\top})u_{\bullet k}|\\
        &\leq | u_{\bullet k}^\top \Sigma_{\bullet 1} \Upsilon^\top u_{\bullet k} | 
        + | u_{\bullet k}^\top \Sigma_{\bullet 1} \Theta^\top u_{\bullet k} | 
        + | u_{\bullet k}^\top \Upsilon \Sigma_{\bullet 1}^\top u_{\bullet k} | 
        + | u_{\bullet k}^\top \Upsilon \Upsilon^\top u_{\bullet k} | \\
        &\quad 
        + | u_{\bullet k}^\top \Upsilon \Theta^\top u_{\bullet k} | 
        + | u_{\bullet k}^\top \Theta \Sigma_{\bullet 1}^\top u_{\bullet k} | + \left| u_{\bullet k}^\top \Theta \Upsilon^\top u_{\bullet k} \right| 
        + | u_{\bullet k}^\top \Theta \Theta^\top u_{\bullet k} |\\
        & = O_{P}\left(\frac{\sqrt{n_{1}D}}{m^{\frac{1}{8}}} + \frac{n_{1}\sqrt{\varphi_{D}}}{m^{\frac{1}{8}}} + n_{1}D\rho_{m}\right).
    \end{align*}

    Note that, for each $k\leq r$ and $j\leq n$, $\Var(\be_{\cdot j}^{\top}u_{\bullet k}) =  u_{\bullet k}^{\top}\bOmega_{e_j,D\times D}u_{\bullet k} \leq \|\bOmega_{e_j,D \times D}\|\|u_{\bullet k}\|^2= O(\varphi_{D})$ by using Assumption \ref{assum_spotvol}(v).
    Then, we have 
    \begin{align}
        &\Var(v_{\bullet k}^{\top}E_{\bullet 1}^{\top}u_{\bullet k}) = \sum_{j=1}^{n_1}v_{\bullet, jk}^2\Var(\be_{\cdot j}^{\top}u_{\bullet k}) + 2\sum_{1\leq j<l\leq n_1}v_{\bullet, jk}v_{\bullet, lk}\cov(\be_{\cdot j}^{\top}u_{\bullet k},\be_{\cdot l}^{\top}u_{\bullet k})\nonumber\\
        &\qquad\qquad\leq \sum_{j=1}^{n_1}v_{\bullet, jk}^2\Var(\be_{\cdot j}^{\top}u_{\bullet k}) + \sum_{j=1}^{n_1}\sum_{l=1}^{n_1}|v_{\bullet, jk}||v_{\bullet, lk}||\cov(\be_{\cdot j}^{\top}u_{\bullet k},\be_{\cdot l}^{\top}u_{\bullet k})|\nonumber\\
        &\qquad\qquad\leq \sum_{j=1}^{n_1}v_{\bullet, jk}^2\Var(\be_{\cdot j}^{\top}u_{\bullet k}) + 2C\varphi_{D}\sum_{h=1}^{n_{1}}\rho(h)\sum_{j=1}^{n_1-h}|v_{\bullet,jk}||v_{\bullet, (j+h)k}|= O(\varphi_{D}),\label{variance_vEu}
    \end{align}
    where the Cauchy-Schwarz inequality and Assumption \ref{assum_spotvol}(vi) is used for the last equality. 
    Since $E(v_{\bullet k}^{\top}E_{\bullet 1}^{\top}u_{\bullet k}) = 0$, we can obtain $|v_{\bullet k}^{\top}E_{\bullet 1}^{\top}u_{\bullet k}| = O_{P}(\sqrt{\varphi_{D}})$ by Chebyshev's inequality.
    We note that $|u_{\bullet k}^{\top}E_{\bullet 1}E_{\bullet 1}^{\top}u_{\bullet k}| = \|u_{\bullet k}^{\top}E_{\bullet 1}\|_{2}^{2} = O_{P}(n_{1}\varphi_{D})$.
    Therefore, we have
    \begin{align*}
        III &  \leq |u_{\bullet k}^{\top}U_{\bullet 1}\Lambda_{\bullet 1}V_{\bullet 1}^{\top}E_{\bullet 1}^{\top}u_{\bullet k}| + |u_{\bullet k}^{\top}E_{\bullet 1}V_{\bullet 1}\Lambda_{\bullet 1}U_{\bullet 1}^{\top}u_{\bullet k}| + |u_{\bullet k}^{\top}E_{\bullet 1}E_{\bullet 1}^{\top}u_{\bullet k}|\\
        & = |\lambda_{\bullet k}v_{\bullet k}^{\top}E_{\bullet 1}^{\top}u_{\bullet k}| + |u_{\bullet k}^{\top}E_{\bullet 1}v_{\bullet k}\lambda_{\bullet k}| + |u_{\bullet k}^{\top}E_{\bullet 1}E_{\bullet 1}^{\top}u_{\bullet k}| = O_{P}(\sqrt{Dn_{1}\varphi_{D}} + n_{1}\varphi_{D}).
    \end{align*}
    Combining the terms $I$, $II$ and $III$ together, for $k\leq r$, we have
    \begin{align}
          |\hat{\lambda}_{\bullet k}^{2} - \lambda_{\bullet k}^2| & = O_{P}\Bigg(D +n_{1}D\rho_{m} + \sqrt{Dn_{1}\varphi_{D}} + n_{1}\varphi_{D}\Bigg).\label{sinval_sqaure}
    \end{align}
        Therefore, for $k\leq r$, we have
        \begin{align*}
        |\hat{\lambda}_{\bullet k} - \lambda_{\bullet k}| &\leq \frac{|\hat{\lambda}_{\bullet k}^{2} - \lambda_{\bullet k}^2|}{\hat{\lambda}_{\bullet k} + \lambda_{\bullet k}} = O_{P}\Bigg(\sqrt{\frac{D}{n_{1}}} + \sqrt{n_{1}D}\rho_{m} + \sqrt{\varphi_{D}} + \sqrt{\frac{n_{1}}{D}}\varphi_D\Bigg).
        \end{align*}

    Consider (ii).
    Using the similar proof of \eqref{fro_norm_sample_cov}, we can obtain
    \begin{align}\label{fro_norm_sample_cov2}
        E\|\hat{\Sigma}_{1 \bullet }^{\top}\hat{\Sigma}_{1\bullet } - \Sigma_{1 \bullet}^{\top}\Sigma_{1\bullet}\|_{F}^{2} = O(Dn^{2}m^{-\frac{1}{4}}+ D^{2}nm^{-\frac{1}{4}}+ D^{2}n^{2}m^{-\frac{1}{8}}\rho_{m} + D^{2}n^2\rho_{m}^{2}). 
    \end{align}
    Let $\Lambda_{1 \bullet} = \Diag(\lambda_{1 \bullet}, \dots,\lambda_{r \bullet})$.
    Using the similar proof of \eqref{sinval_sqaure}, for $k\leq r$, we can obtain
    \begin{align*}
          |\hat{\lambda}_{k \bullet}^{2} - \lambda_{k \bullet }^2| & = O_{P}\Bigg(n+ nD\rho_{m}  + \sqrt{Dn\varphi_{n}} + D\varphi_{n}\Bigg).
    \end{align*}
        Therefore, for $k\leq r$, we have
        \begin{align*}        
        |\hat{\lambda}_{k \bullet} - \lambda_{k \bullet}| &\leq \frac{|\hat{\lambda}_{k \bullet}^{2} - \lambda_{k \bullet}^{2}|}{\hat{\lambda}_{k \bullet} + \lambda_{k \bullet}} = O_{P}\Bigg(\sqrt{{\frac{n}{D}}} +\sqrt{nD}\rho_{m}  + \sqrt{\varphi_{n}} + \sqrt{\frac{D}{n}}\varphi_n\Bigg).
        \end{align*}
    \end{proof}

We define $U_{\bullet 1}$ and $V_{\bullet 1}$ as the left and right singular vector matrices of  $A_{\bullet 1} \equiv [A_{11}^{\top} A_{21}^{\top}]^{\top}$, and $U_{1 \bullet}$ and $V_{1 \bullet}$ as the left and right singular vector matrices of $A_{1\bullet} \equiv [A_{11} A_{12}]$.
The columns of $\hat{U}_{\bullet 1}$ and $\hat{V}_{\bullet 1}$ are defined as the top $r$ eigenvectors of $\hat{\Sigma}_{\bullet 1}\hat{\Sigma}_{\bullet 1}^{\top}$ and $\hat{\Sigma}_{\bullet 1}^{\top}\hat{\Sigma}_{\bullet 1}$, respectively.
Similarly, the columns of $\hat{U}_{1\bullet }$ and $\hat{V}_{1 \bullet}$ are defined as the top $r$ eigenvectors of $\hat{\Sigma}_{1\bullet}\hat{\Sigma}_{1\bullet}^{\top}$ and $\hat{\Sigma}_{1\bullet}^{\top}\hat{\Sigma}_{1\bullet}$, respectively.

The following lemma presents the individual convergence rate of the singular vector estimators, up to sign. 
We note that the SIP procedure does not require sign alignment; therefore, we focus on the convergence rates of the singular vector estimators, which are equivalent to those of the eigenvector estimators based on the sample covariance matrices.

\begin{lem}\label{sin_vectors}
        Under Assumption \ref{assum_spotvol}, if $\log D = o(n_{1})$ and $\log n = o(D)$, we have the following results:
        \begin{align*}
            \text{(i) } &\|\hat{U}_{\bullet 1} - U_{\bullet 1}\|_{\max} = O_{P}\left(\frac{\rho_{m}^{2}}{\sqrt{D}} + \frac{1}{\sqrt{D} m^{\frac{1}{4}}} + \sqrt{\frac{\log D}{n_{1}D}} + \frac{\varphi_D}{D\sqrt{D}}\right),\\
            \text{(ii) } &\|\hat{V}_{1 \bullet} -V_{1 \bullet}\|_{\max} = O_{P}\left(\frac{\rho_{m}^{2}}{\sqrt{n}}+ \frac{1}{\sqrt{n}m^{\frac{1}{4}}}+ \sqrt{\frac{\log n}{nD}} + \frac{\varphi_n}{n\sqrt{n}}\right),\\
            \text{(iii) } &\|\hat{U}_{1\bullet} - U_{1\bullet}\|_{\max} = O_{P}\left(\frac{\rho_{m}^{2}}{\sqrt{D}} + \frac{1}{\sqrt{D} m^{\frac{1}{4}}} + \sqrt{\frac{\log D}{nD}} + \frac{\varphi_D}{D\sqrt{D}}\right),\\
            \text{(iv) } &\|\hat{V}_{\bullet 1} -V_{\bullet 1}\|_{\max} = O_{P}\left(\frac{\rho_{m}^{2}
            }{\sqrt{n_{1}}}+ \frac{1}{\sqrt{n_{1}}m^{\frac{1}{4}}} + \sqrt{\frac{\log n_{1}}{n_{1}D}} + \frac{\varphi_n}{n_{1}\sqrt{n_{1}}}\right).
            \end{align*}
    \end{lem}
    \begin{proof}
    We first consider $(i)$.
    We have
    \begin{align*}
        \|\hat{\Sigma}_{\bullet 1}\hat{\Sigma}_{\bullet 1}^{\top}-\Sigma_{\bullet 1}\Sigma_{\bullet 1}^{\top}\|_{\max} &= \max_{i,s\leq D}|\sum_{j=1}^{n_{1}}(\hat{c}_{i,j}\hat{c}_{s,j}-c_{i,j}c_{s,j})|\\
        &  = \max_{i,s\leq D}|\sum_{j=1}^{n_{1}}(\upsilon_{i,j}c_{s,j} + \upsilon_{s,j}c_{i,j})| + \max_{i,s\leq D}|\sum_{j=1}^{n_{1}}(\varsigma_{i,j}c_{s,j} + \varsigma_{s,j}c_{i,j})| \\
        & \qquad+ \max_{i,s\leq D}|\sum_{j=1}^{n_{1}}(\upsilon_{i,j}\upsilon_{s,j} + \varsigma_{i,j}\varsigma_{s,j} + \upsilon_{i,j}\varsigma_{s,j} + \varsigma_{i,j}\upsilon_{s,j})|\\
        & := I + II + III.
    \end{align*}
    For each $(i,s)$, define $Z_{j}^{(i,s)} = \upsilon_{i,j}c_{s,j}$.
    Recall that both $\upsilon_{i,j}$ and $c_{s,j}$ are sub-Gaussian random variables by Assumption \ref{assum_spotvol}(iii).
    Then, $Z_{j}^{(i,s)}$ is sub-exponential, satisfying
    $$
    \|Z_{j}^{(i,s)}\|_{\psi_{1}} \leq C\|\upsilon_{i,j}\|_{\psi_{2}}\|c_{s,j}\|_{\psi_{2}} = O(m^{-1/8}):=\nu.
    $$
    Since $Z_{j}^{(i,s)}$ forms a martingale difference sequence, we apply a Freedman-type inequality: 
    $$    P\left(\left|\sum_{j=1}^{n_{1}}Z_{j}^{(i,s)}\right| >t \right) \leq 2\exp\left(-c\cdot \min\left(\frac{t^2}{n_{1}\nu^2},\frac{t}{\nu}\right)\right). 
    $$
    Take $t= C m^{-1/8} \sqrt{n_1 \log D}$. 
    Then, we have $\frac{t^2}{n_{1}\nu^2} = O(\log D)$ and $\frac{t}{\nu} = O(\sqrt{n_{1} \log D})$.
    Under the condition $\log D = o(n_{1})$, we have $\min(\frac{t^2}{n_{1}\nu^2},\frac{t}{\nu}) = \frac{t^2}{n_{1}\nu^2}$.
    Then, for some large constant $C>0$, we have
    $$    P\left(\max_{i,s}\left|\sum_{j=1}^{n_{1}}Z_{j}^{(i,s)}\right| >t \right) \leq 2D^{2}\exp\left(-c\cdot \frac{t^2}{n_{1}\nu^2}\right) \leq D^{-1}. 
    $$
    Therefore, we have
    $$
    I \leq 2\cdot \max_{i,s\leq D}\left|\sum_{j=1}^{n_{1}}\upsilon_{i,j}c_{s,j}\right| = O_{P}(m^{-1/8}\sqrt{n_{1}\log D}).
    $$
    Similarly, define $W_{j}^{(i,s)} = \varsigma_{i,j}c_{s,j}$ for each $(i,s)$. 
    Under Assumption \ref{assum_spotvol}(iii),  we have 
    $$
    \|W_{j}^{(i,s)}\|_{\psi_{1}} \leq C\|\varsigma_{i,j}\|_{\psi_{2}}\|c_{s,j}\|_{\psi_{2}} = O(\rho_{m}):=\mu.
    $$
    For each fixed $(i,s)$, we apply Bernstein's inequality:
    $$
    P\left(\left|\sum_{j=1}^{n_{1}}W_{j}^{(i,s)}\right| >t \right) \leq 2\exp\left(-c\cdot \min\left(\frac{t^2}{n_{1}\mu^2},\frac{t}{\mu}\right)\right). 
    $$
    Take $t = C\rho_{m}\sqrt{n_1 \log D}$.
    Since $\log D = o(n_{1})$, it follows that $\min(\frac{t^2}{n_{1}\mu^2},\frac{t}{\mu}) = \frac{t^2}{n_{1}\mu^2}$.
    Then, for some large constant $C>0$, we have
    $$
    P\left(\max_{i,s}\left|\sum_{j=1}^{n_{1}}W_{j}^{(i,s)}\right| >t \right) \leq 2D^{2}\exp\left(-c\cdot \frac{t^2}{n_{1}\mu^2}\right)\leq D^{-1}. 
    $$
    Therefore, we have
    $$
    II \leq 2\cdot \max_{i,s\leq D}\left|\sum_{j=1}^{n_{1}}\varsigma_{i,j}c_{s,j}\right| = O_{P}(\rho_{m}\sqrt{n_{1}\log D}).
    $$
    To bound $III$, we handle the off-diagonal ($i\neq s$) and diagonal ($i =  s$) cases separately.
    We note that for $i\neq s$,
    $E[\upsilon_{i,j}\upsilon_{s,j}]=0$ by Assumption \ref{assum_spotvol}(ii).
    Since $\upsilon_{i,j}$ is sub-Gaussian, we have
    $$
    \|\upsilon_{i,j}\upsilon_{s,j}\|_{\psi_{1}} \leq C\|\upsilon_{i,j}\|_{\psi_{2}}\|\upsilon_{s,j}\|_{\psi_{2}} = O(m^{-1/4}):=\eta_{1}.
    $$
    For each fixed $(i,s)$, we apply Bernstein's inequality for sub-exponential variables:
    $$
    P\left(\left|\sum_{j=1}^{n_{1}}\upsilon_{i,j}\upsilon_{s,j}\right| >t \right) \leq 2\exp\left(-c\cdot \min\left(\frac{t^2}{n_{1}\eta_{1}^2},\frac{t}{\eta_{1}}\right)\right). 
    $$
    Take $t = Cm^{-1/4}\sqrt{n_1 \log D}$.
    Since $\log D = o(n_{1})$, it follows that $\min(\frac{t^2}{n_{1}\eta_{1}^2},\frac{t}{\eta_{1}}) = \frac{t^2}{n_{1}\eta_{1}^2}$.
    Then, for some large constant $C>0$, we have
    $$
    P\left(\max_{i\neq s}\left|\sum_{j=1}^{n_{1}}\upsilon_{i,j}\upsilon_{s,j}\right| >t \right) \leq 2D^{2}\exp\left(-c\cdot \frac{t^2}{n_{1}\eta_{1}^2}\right) \leq D^{-1}.     
    $$
    Therefore, we have
    $$
    \max_{i\neq s}\left|\sum_{j=1}^{n_{1}}\upsilon_{i,j}\upsilon_{s,j}\right| = O(m^{-1/4}\sqrt{n_{1}\log D}).
    $$
    We also note that 
    $$    
    E(\sum_{j=1}^{n_{1}}\upsilon_{i,j}^2) =\sum_{j=1}^{n_{1}}E(\upsilon_{i,j}^2)= O(n_1m^{-1/4}). 
    $$
    For each $i$, we apply Bernstein's inequality:
    $$
    P\left(\left|\sum_{j=1}^{n_{1}}\upsilon_{i,j}^{2}-E\left[\sum_{j=1}^{n_{1}}\upsilon_{i,j}^{2}\right]\right|>t\right) \leq 2 \exp\left(-c\cdot \min\left(\frac{t^2}{n_{1}\eta_{2}^2},\frac{t}{\eta_{2}}\right)\right),
    $$
    where $\eta_{2}:=\|\upsilon_{i,j}^{2}\|_{\psi_{1}} = O(m^{-1/4})$.
    Take $t = Cm^{-1/4}\sqrt{n_1 \log D}$.
    Since $\log D = o(n_{1})$, it follows that $\min(\frac{t^2}{n_{1}\eta_{2}^2},\frac{t}{\eta_{2}}) =  \frac{t^2}{n_{1}\eta_{1}^2}$.
    Then, for some large constant $C>0$, we have 
    $$
    P\left(\max_{i\leq D}\left|\sum_{j=1}^{n_{1}}\upsilon_{i,j}^{2}-E\left[\sum_{j=1}^{n_{1}}\upsilon_{i,j}^{2}\right]\right|>t\right) \leq 2D \exp\left(-c\cdot \frac{t^2}{n_{1}\eta_{2}^2}\right)\leq D^{-1}.
    $$
    Therefore, we have
    $$
    \max_{i\leq D}\left|\sum_{j=1}^{n_{1}}\upsilon_{i,j}^{2}-E\left[\sum_{j=1}^{n_{1}}\upsilon_{i,j}^{2}\right]\right| = O_{P}(m^{-1/4}\sqrt{n_{1}\log D}),
    $$
    and thus,
    $$
    \max_{i\leq D}\left|\sum_{j=1}^{n_{1}}\upsilon_{i,j}^{2}\right| = O_{P}(m^{-1/4}(n_{1} + \sqrt{n_{1}\log D}))= O_{P}(m^{-1/4}n_{1}).
    $$
    By combining the diagonal and off-diagonal terms, we obtain
    \begin{align}\label{upsilon_cross}
    \max_{i,s\leq D}|\sum_{j=1}^{n_{1}}\upsilon_{i,j}\upsilon_{s,j}| = O_{P}(m^{-1/4}n_{1}).    
    \end{align}
    By using the sub-Gaussianity of $\varsigma_{i,j}$ and the same argument as in \eqref{upsilon_cross},  we also obtain
    $$
    \max_{i,s\leq D}|\sum_{j=1}^{n_{1}}\varsigma_{i,j}\varsigma_{s,j}| = O_{P}(\rho_{m}^{2}n_{1}).
    $$
    Lastly, define  $Q_{j}^{(i,s)} = \upsilon_{i,j}\varsigma_{s,j}$ for each $(i,s)$. 
    Then, we have
    $$
    \|Q_{j}^{(i,s)}\|_{\psi_{1}} \leq C\|\upsilon_{i,j}\|_{\psi_{2}}\|\varsigma_{s,j}\|_{\psi_{2}} = O(m^{-1/8}\rho_{m}):=\pi.
    $$
    For each fixed $(i,s)$, we apply a Bernstein's inequality:
    $$
    P\left(\left|\sum_{j=1}^{n_{1}}Q_{j}^{(i,s)}\right| >t \right) \leq 2\exp\left(-c\cdot \min\left(\frac{t^2}{n_{1}\pi^2},\frac{t}{\pi}\right)\right). 
    $$
    Take $t = Cm^{-1/8}\rho_{m}\sqrt{n_{1}\log D}$.
    Since $\log D = o(n_{1})$, it follows that $\min(\frac{t^2}{n_{1}\pi^2},\frac{t}{\pi}) = \frac{t^2}{n_{1}\pi^2}$.
    Then, for some large constant $C>0$, we have
    $$  P\left(\max_{i,s}\left|\sum_{j=1}^{n_{1}}\upsilon_{i,j}\varsigma_{s,j}\right| >t \right) \leq 2D^{2}\exp\left(-c\cdot \frac{t^2}{n_{1}\pi^2}\right) \leq D^{-1}. 
    $$
    Therefore, we have
    $$
    \max_{i,s}\left|\sum_{j=1}^{n_{1}}\upsilon_{i,j}\varsigma_{s,j}\right| = O_{P}(m^{-1/8}\rho_{m}\sqrt{n_{1}\log D}).
    $$
    Hence, we obtain
    \begin{align*}
    III &\leq \max_{i,s\leq D}|\sum_{j=1}^{n_{1}}\upsilon_{i,j}\upsilon_{s,j}| + \max_{i,s\leq D}|\sum_{j=1}^{n_{1}}\varsigma_{i,j}\varsigma_{s,j}| + 2\cdot \max_{i,s\leq D}|\sum_{j=1}^{n_{1}}\upsilon_{i,j}\varsigma_{s,j}|\\
    & = O_{P}(m^{-1/4}n_{1} + \rho_{m}^{2}n_{1} +m^{-1/8}\rho_{m}\sqrt{n_{1}\log D}).       
    \end{align*}
    Combining the bounds for terms $I,II$ and $III$ together, we obtain
    \begin{equation}\label{maxnorm}
     \|\hat{\Sigma}_{\bullet 1}\hat{\Sigma}_{\bullet 1}^{\top}-\Sigma_{\bullet 1}\Sigma_{\bullet 1}^{\top}\|_{\max} = O_{P}\left(m^{-\frac{1}{8}}\sqrt{n_{1}\log D} + \rho_{m}\sqrt{n_{1}\log D}+m^{-\frac{1}{4}}n_{1} + \rho_{m}^2 n_{1}\right).       \end{equation}

        Let $\hat{U}_{\bullet 1} = (\hat{u}_{\bullet 1},\ldots, \hat{u}_{\bullet r})$, $U_{\bullet 1} = (u_{\bullet 1},\ldots, u_{\bullet r})$, and $\Lambda_{\bullet 1} = \Diag(\lambda_{\bullet 1}, \dots,\lambda_{\bullet r})$.
        We denote the eigengap $\bar{\gamma} = \min\{\lambda_{\bullet k}^{2} - \lambda_{\bullet k+1}^{2} : 1 \leq k \leq r\}$ and  $\lambda_{\bullet r+1} = 0$. For $k\leq r$, $\lambda_{\bullet k} \asymp \sqrt{n_{1}D}$  and  $\|u_{\bullet k}\|_{\infty} \leq C/\sqrt{D}$. 
		In addition, the coherence $\mu(U_{\bullet 1}) = D\max_{i}\sum_{k=1}^{r}u_{\bullet,ik}^{2}/r \leq C$, where $u_{\bullet,ik}$ is the $(i,k)$ entry of $U_{\bullet 1}$.
		Thus, by Theorem 1 of \cite{fan2018eigenvector}, we have
        \begin{align*}
            &\max_{k\leq r} \|\hat{u}_{\bullet k}-u_{\bullet k}\|_{\infty}       \leq C\frac{\|\hat{\Sigma}_{\bullet 1}\hat{\Sigma}_{\bullet 1}^{\top}-U_{\bullet 1}\Lambda_{\bullet 1}^{2}U_{\bullet 1}^{\top}\|_{\infty}}{\bar{\gamma}\sqrt{D}}\\
            &\quad \leq C\frac{\|\hat{\Sigma}_{\bullet 1}\hat{\Sigma}_{\bullet 1}^{\top}-\Sigma_{\bullet 1}\Sigma_{\bullet 1}^{\top}\|_{\infty} + \|\Sigma_{\bullet 1}\Sigma_{\bullet 1}^{\top}-\sum_{j=1}^{n_{1}}\bOmega_{j,D\times D}\|_{\infty}+ \|\sum_{j=1}^{n_{1}}\bOmega_{j, D\times D}-U_{\bullet 1}\Lambda_{\bullet 1}^{2}U_{\bullet 1}^{\top}\|_{\infty}}{\bar{\gamma}\sqrt{D}}\\
            &\quad = O\left(\frac{D\|\hat{\Sigma}_{\bullet 1}\hat{\Sigma}_{\bullet 1}^{\top}-\Sigma_{\bullet 1}\Sigma_{\bullet 1}^{\top}\|_{\max}}{n_{1}D^{\frac{3}{2}}} + \frac{D\|\Sigma_{\bullet 1}\Sigma_{\bullet 1}^{\top}-\sum_{j=1}^{n_{1}}\bOmega_{j,D\times D}\|_{\max}}{n_{1}D^{\frac{3}{2}}} + \frac{\sum_{j=1}^{n_{1}}\|\bOmega_{e_j,D\times D}\|_{1}}{n_{1}D^{\frac{3}{2}}}\right)\\
            &\quad =  O_{P}\left(\sqrt{\frac{\log D}{n_{1}Dm^{\frac{1}{4}}}}+ \rho_{m}\sqrt{\frac{\log D}{n_{1}D}} + \frac{1}{\sqrt{D} m^{\frac{1}{4}}} + \frac{\rho_{m}^{2}}{\sqrt{D}} \right)+  O_{P}\left(\sqrt{\frac{\log D}{n_{1}D}}\right) + O_{P}\left(\frac{\varphi_D}{D\sqrt{D}}\right)            \\
            &\quad =  O_{P}\left( \frac{1}{\sqrt{D} m^{\frac{1}{4}}} + \frac{\rho_{m}^{2}}{\sqrt{D}} + \sqrt{\frac{\log D}{n_{1}D}} + \frac{\varphi_D}{D\sqrt{D}}\right),
        \end{align*}   
        where the second equality follows from \eqref{maxnorm} and Assumptions \ref{assum_spotvol}(iv)--(v).
    	By using a similar argument, we can obtain the results (ii)--(iv).

    \end{proof}

 The following lemma presents the individual convergence rate of singular value estimators  for the partitioned matrix $A_{11}$.

\begin{lem}\label{singular value}
Under Assumption \ref{assum_spotvol}, if $\log D = o(n_{1})$, $\log n = o(D)$, $\varphi_{n}\sqrt{D} \leq n_{1}$ and $\varphi_{D} \sqrt{n_{1}} \leq D$, we have
    \begin{align*}
        &\|\hat{U}_{11}^{\top}\hat{\Sigma}_{11}\hat{V}_{11} - \Lambda_{11}\|_{\max} \\
        &\qquad= O_{P}\Bigg(\rho_{m}\sqrt{Dn_{1}} + \frac{\sqrt{Dn_{1}}}{m^{\frac{1}{4}}} + \sqrt{D\log D} + \sqrt{\frac{n_{1}}{D}}\varphi_{D} +  \sqrt{n_{1}\log n_{1}} + \sqrt{\frac{D}{n_{1}}}\varphi_{n}\Bigg).       
    \end{align*}
\end{lem}
\begin{proof}

   Since $\Lambda_{11} = U_{11}^{\top}A_{11}V_{11}$, we have 
\begin{align*}
        \|\hat{U}_{11}^{\top}\hat{\Sigma}_{11}\hat{V}_{11} - \Lambda_{11}\|_{\max}   &= \|\hat{U}_{11}^{\top}\hat{\Sigma}_{11}\hat{V}_{11} - U_{11}^{\top}A_{11}V_{11}\|_{\max}\\
        & \leq \|\hat{U}_{11}^{\top}(\hat{\Sigma}_{11}-A_{11})\hat{V}_{11}\|_{\max} + \|(\hat{U}_{11}-U_{11})^{\top}A_{11}(\hat{V}_{11}-V_{11})\|_{\max} \\
        & \qquad + \|(\hat{U}_{11}-U_{11})^{\top}A_{11}V_{11}\|_{\max} + \|U_{11}^{\top}A_{11}(\hat{V}_{11}-V_{11})\|_{\max} \\
        & := \Delta_{1} + \Delta_{2} + \Delta_{3} + \Delta_{4}.
    \end{align*}
In order to find the convergence rate of elementwise norm, we consider the above terms separately. 
For $\Delta_{1}$, we have   
   \begin{align*}
        \Delta_{1}  &= \|\hat{U}_{11}^{\top}(\hat{\Sigma}_{11}-\Sigma_{11} + \Sigma_{11}-A_{11})\hat{V}_{11}\|_{\max}\\
        & \leq \|\hat{U}_{11}^{\top}(\hat{\Sigma}_{11}-\Sigma_{11})\hat{V}_{11}\|_{\max} + \|\hat{U}_{11}^{\top}E_{11}\hat{V}_{11}\|_{\max}\\
        &  \leq \|\hat{U}_{11}^{\top}(\hat{\Sigma}_{11}-\Sigma_{11})\hat{V}_{11}\|_{\max} + \|U_{11}^{\top}E_{11}V_{11}\|_{\max} + \|(\hat{U}_{11} - U_{11})^{\top}E_{11}(\hat{V}_{11} - V_{11})\|_{\max}\\
        &  \qquad+ \|(\hat{U}_{11}-U_{11})^{\top}E_{11}V_{11}\|_{\max} + \|U_{11}^{\top}E_{11}(\hat{V}_{11}-V_{11})\|_{\max}\\
        &  := I + II + III+IV+V. 
    \end{align*}
Let $U_{11} = (u_{i,k})_{(D-1) \times r}$ and $V_{11} = (v_{j,k})_{n_1\times r}$. 
For each $(k,l)$, define $S_{kl} = \sum_{i=1}^{D-1}\sum_{j=1}^{n_1}u_{i,k}\upsilon_{i,j}v_{j,l}$.
We note that, by Assumption \ref{assum_spotvol}(ii)--(iii), $E(\upsilon_{i,j}\upsilon_{i',j'}) =0$ unless $i=i'$ and $j=j'$.
Then, for any $k,l\leq r$, we have 
$$
E(S_{kl}) = \sum_{i=1}^{D-1}\sum_{j=1}^{n_1}u_{i,k}v_{j,l}E(\upsilon_{i,j}) = 0
$$
and
$$
  \Var(S_{kl}) =  \sum_{i=1}^{D-1}\sum_{j=1}^{n_{1}}u_{i,k}^{2}v_{j,l}^{2}E[\upsilon_{i,j}^2] = O(m^{-1/4}).
$$
Then, by Chebyshev's inequality, we have
$$
|\sum_{i=1}^{D-1}\sum_{j=1}^{n_1}u_{i,k}\upsilon_{i,j}v_{j,l}| = O_{P}(m^{-1/8}).
$$
By Assumption \ref{assum_spotvol}(ii)--(iii), we have
$$
|\sum_{i=1}^{D-1}\sum_{j=1}^{n_1}u_{i,k}\varsigma_{i,j}v_{j,l}| \leq   \sum_{i=1}^{D-1}\sum_{j=1}^{n_1} |u_{i,k}| |\varsigma_{i,j}| |v_{j,l}| = O_{P}(\rho_{m}\sqrt{Dn_{1}}).
$$
Since $r$ is fixed, we have
\begin{align*}
\|U_{11}^{\top}(\hat{\Sigma}_{11}-\Sigma_{11})V_{11}\|_{\max} &\leq \max_{k,l\leq r}|\sum_{i=1}^{D-1}\sum_{j=1}^{n_1}u_{i,k}\upsilon_{i,j}v_{j,l}| + \max_{k,l\leq r}|\sum_{i=1}^{D-1}\sum_{j=1}^{n_1}u_{i,k}\varsigma_{i,j}v_{j,l}|\\
&= O_{P}(m^{-1/8} + \rho_{m}\sqrt{Dn_{1}}).   
\end{align*}
Therefore, we have 
$$
I = O_{P}(\|U_{11}^{\top}(\hat{\Sigma}_{11}-\Sigma_{11})V_{11}\|_{\max}) = O_{P}(m^{-1/8} + \rho_{m}\sqrt{Dn_{1}}).
$$
Using the same argument as \eqref{variance_vEu}, we have 
$$
\Var(\sum_{i=1}^{D-1}\sum_{j=1}^{n_{1}}u_{i,k}v_{j,l}e_{i,j}) = O(\varphi_{D})
$$ 
and hence, 
$$|\sum_{i=1}^{D-1}\sum_{j=1}^{n_{1}}u_{i,k}v_{j,l}e_{i,j}| = O_{P}(\sqrt{\varphi_{D}}).
$$
Since $r$ is fixed, we can apply a union bound over $r^2$ terms without affecting the rate.
Therefore, we have
$$
II  = O(\sqrt{\varphi_{D}}).
$$
We note that $\hat{U}_{11}$ and $\hat{U}_{1\bullet}$ are equal up to column-wise sign changes; the same holds for $\hat{V}_{11}$ and $\hat{V}_{\bullet 1}$.
We define $Q_{k,j} = \sum_{i=1}^{D-1}u_{i,k}e_{i,j}$.
We note that 
$$
\Var(Q_{k,j}) = u_{k}^{\top}\bOmega_{e_{j},D\times D}u_{k} \leq \max_{j\leq n}\|\bOmega_{e_{j},D\times D}\| \|u_{k}\|^2 =  O(\varphi_{D}).
$$
Using the bounded forth moment condition in Assumption \ref{assum_spotvol}(vi) and Markov's inequality, we have 
$$
P(|Q_{k,j}|>t) \leq \frac{C\varphi_{D}^2}{t^{4}}.
$$
By applying the union bound over all $(k,j)$ pairs, we then have
$$
P(\max_{k,j}|Q_{k,j}| > t) \leq Crn_{1}\cdot\frac{\varphi_{D}^2}{t^4}.
$$
Therefore, we have
$$
\|U_{11}^{\top}E_{11}\|_{\max} = O_{P}(\varphi_{D}^{1/2}n_{1}^{1/4}) := O_{P}(\kappa_{u}).
$$ 
Similarly, we have $\|E_{11}V_{11}\|_{\max} = O_{P}(\varphi_{n}^{1/2}D^{1/4}) := O_{P}(\kappa_{v})$.
Then, by Lemma \ref{sin_vectors} (iii)--(iv), we can obtain
\begin{align*}
    IV & = O_{P}(\kappa_{v}D\|\hat{U}_{1\bullet}-U_{1\bullet}\|_{\max}) \\
        &= O_{P}\left(\kappa_{v}\sqrt{D}\left(\rho_{m}^2+\frac{1}{m^{\frac{1}{4}}}\right) + \kappa_{v}\sqrt{\frac{D\log D}{n}} + \frac{\kappa_{v}\varphi_D}{\sqrt{D}}\right),\\
      V & = O_{P}(\kappa_{u}n_{1} \|\hat{V}_{\bullet 1}-V_{\bullet 1}\|_{\max})\\
        &= O_{P}\left(\kappa_{u}\sqrt{n_{1}}\left(\rho_{m}^2+\frac{1}{m^{\frac{1}{4}}}\right)+ \kappa_{u}\sqrt{\frac{n_{1}\log n_{1}}{D}} + \frac{\kappa_{u}\varphi_n}{\sqrt{n_{1}}}\right).
\end{align*}
Similarly, we can show that the term $III$ is dominated by $IV$ and $V$.
Therefore, we have
\begin{align*}
\Delta_{1} &= O_{P}\Bigg(\rho_{m}\sqrt{Dn_{1}} + \sqrt{\varphi_{D}} + \kappa_{v}\sqrt{D}\left(\rho_{m}^2+\frac{1}{m^{\frac{1}{4}}}\right) + \kappa_{v}\sqrt{\frac{D\log D}{n}} + \frac{\kappa_{v}\varphi_D}{\sqrt{D}}\\
& \qquad\qquad +\kappa_{u}\sqrt{n_{1}}\left(\rho_{m}^2+\frac{1}{m^{\frac{1}{4}}}\right)+ \kappa_{u}\sqrt{\frac{n_{1}\log n_{1}}{D}} + \frac{\kappa_{u}\varphi_n}{\sqrt{n_{1}}}\Bigg).
\end{align*}
For $\Delta_3$ and $\Delta_4$, by Lemma \ref{sin_vectors} (iii)--(iv), we have
\begin{align*}
        \Delta_3 & = \|(\hat{U}_{11}-U_{11})^{\top}U_{11}\Lambda_{11}\|_{\max} = O_{P}(D \sqrt{n_{1}}\|\hat{U}_{1\bullet}-U_{1\bullet}\|_{\max})\\
        & = O_{P}\left(\sqrt{n_{1}D}\left(\rho_{m}^2+\frac{1}{m^{\frac{1}{4}}}\right) + \sqrt{\frac{n_{1}D\log D}{n}} + \varphi_D\sqrt{\frac{n_{1}}{D}}\right),\\
        \Delta_4 & = \|\Lambda_{11}V_{11}^{\top}(\hat{V}_{11}-V_{11})\|_{\max} = O_{P}(n_{1}\sqrt{D}\|\hat{V}_{\bullet 1} - V_{\bullet 1}\|_{\max})\\
        & = O_{P}\left(\sqrt{n_{1}D}\left(\rho_{m}^2+\frac{1}{m^{\frac{1}{4}}}\right) + \sqrt{n_{1}\log n_{1}} + \varphi_n\sqrt{\frac{D}{n_{1}}}\right).
    \end{align*}
Similarly, we can show $\Delta_{2}$ is dominated by $\Delta_{3}$ and $\Delta_{4}$.
Combining the terms $\Delta_{1}$, $\Delta_{3}$ and $\Delta_{4}$ together, we have
     \begin{align*}
         &\|\hat{U}_{11}^{\top}\hat{\Sigma}_{11}\hat{V}_{11} - \Lambda_{11}\|_{\max} \\
         &\qquad= O_{P}\Bigg(\rho_{m}\sqrt{Dn_{1}} + \sqrt{\varphi_{D}} + \max(\kappa_{v}, \sqrt{n_{1}})\Bigg(\sqrt{D}\left(\rho_{m}^2+\frac{1}{m^{\frac{1}{4}}}\right)+\sqrt{\frac{D\log D}{n}} + \frac{\varphi_D}{\sqrt{D}}\Bigg)\\
        & \qquad\qquad\qquad  + \max(\kappa_{u}, \sqrt{D})\Bigg(\sqrt{n_{1}}\left(\rho_{m}^2+\frac{1}{m^{\frac{1}{4}}}\right) + \sqrt{\frac{n_{1}\log n_{1}}{D}} + \frac{\varphi_n}{\sqrt{n_{1}}}\Bigg)\Bigg)\\
        &\qquad= O_{P}\Bigg(\rho_{m}\sqrt{Dn_{1}} + \frac{\sqrt{Dn_{1}}}{m^{\frac{1}{4}}} + \sqrt{D\log D} + \sqrt{\frac{n_{1}}{D}}\varphi_{D} +  \sqrt{n_{1}\log n_{1}} + \sqrt{\frac{D}{n_{1}}}\varphi_{n}\Bigg).
    \end{align*}

\end{proof}

    \subsection{Proof of Theorem \ref{main_thm}}
We recall that $\tilde{\Sigma}_{22}  =\hat{\Sigma}_{21}\hat{V}_{11}(\hat{U}_{11}^{\top}\hat{\Sigma}_{11}\hat{V}_{11})^{-1}\hat{U}_{11}^{\top}\hat{\Sigma}_{12}= \hat{A}_{21}\hat{V}_{11}(\hat{U}_{11}^{\top}\hat{\Sigma}_{11}\hat{V}_{11})^{-1}\hat{U}_{11}^{\top}\hat{A}_{12}$. 
By Lemmas \ref{sin_val}--\ref{sin_vectors}, we can show
\begin{align*}
        \frac{1}{\sqrt{n_1}}\|\hat{A}_{21}\hat{V}_{11} - A_{21}V_{11}\|_{\max} & \leq \frac{1}{\sqrt{n_1}}\|\hat{A}_{\bullet 1}\hat{V}_{11} - A_{\bullet 1}V_{11}\|_{\max}  = \frac{1}{\sqrt{n_1}}\|\hat{U}_{\bullet 1}\hat{\Lambda}_{\bullet 1} - U_{\bullet 1}\Lambda_{\bullet 1}\|_{\max}\\
        & \leq \frac{1}{\sqrt{n_1}}\|\hat{U}_{\bullet 1}(\hat{\Lambda}_{\bullet 1} - \Lambda_{\bullet 1})\|_{\max} + \frac{1}{\sqrt{n_1}}\|(\hat{U}_{\bullet 1} - U_{\bullet 1})\Lambda_{\bullet 1}\|_{\max}\\
        & = O_{P}\left(\frac{1}{\sqrt{n_{1}D}}\|\hat{\Lambda}_{\bullet 1} - \Lambda_{\bullet 1}\|_{\max} + \sqrt{D}\|\hat{U}_{\bullet 1} - U_{\bullet 1}\|_{\max} \right)\\
        &= O_{P}\left(\rho_{m} +\frac{\varphi_{D}}{D} + \frac{1}{m^{\frac{1}{4}}} +   \sqrt{\frac{\log D}{n_1}}\right)
    \end{align*}
and
    \begin{align*}
        \frac{1}{\sqrt{D}}\|\hat{U}_{11}^{\top}\hat{A}_{12} - U_{11}^{\top}A_{12}\|_{\max} & \leq \frac{1}{\sqrt{D}}\|\hat{U}_{11}^{\top}\hat{A}_{1\bullet} - U_{11}^{\top}A_{1\bullet}\|_{\max} = \frac{1}{\sqrt{D}}\|\hat{\Lambda}_{1\bullet}\hat{V}_{1\bullet}^{\top} - \Lambda_{1\bullet}V_{1\bullet}^{\top}\|_{\max}\\
        & \leq \frac{1}{\sqrt{D}}\|(\hat{\Lambda}_{1\bullet} - \Lambda_{1\bullet})\hat{V}_{1\bullet}^{\top}\|_{\max}+ \frac{1}{\sqrt{D}}\|\Lambda_{1\bullet}(\hat{V}_{1\bullet} -V_{1\bullet})^{\top}\|_{\max} \\
        & = O_{P}\left(\frac{1}{\sqrt{nD}}\|\hat{\Lambda}_{1\bullet } - \Lambda_{1\bullet}\|_{\max} + \sqrt{n}\|\hat{V}_{1\bullet} - V_{1\bullet}\|_{\max} \right)\\
        &= O_{P}\left(\rho_m + \frac{\varphi_{n}}{n}+ \frac{1}{m^{\frac{1}{4}}} + \sqrt{\frac{\log n}{D}}\right).
    \end{align*}
By Lemma \ref{singular value}, we can obtain
\begin{align*}
        &\|(\hat{U}_{11}^{\top}\hat{\Sigma}_{11}\hat{V}_{11})^{-1} - \Lambda_{11}^{-1}\|_{\max}   = O_{P}\left(\frac{1}{n_{1}D}\|\hat{U}_{11}^{\top}\hat{\Sigma}_{11}\hat{V}_{11} - \Lambda_{11}\|_{\max}\right)\\
        & \qquad = O_{P}\Bigg(  \frac{\rho_{m}}{\sqrt{n_{1}D}} +  \frac{1}{\sqrt{n_{1}D}m^{\frac{1}{4}}} + \sqrt{\frac{\log D}{n_{1}^2 D}} + \frac{\varphi_D}{\sqrt{n_{1}D^3}}+ \sqrt{\frac{\log n_{1}}{n_{1}D^2}} + \frac{\varphi_n}{\sqrt{n_{1}^3 D}}\Bigg).   
    \end{align*}
By using the above results, we have
  \begin{align*}
        &\max_{j\leq n_2}\left|\tilde{c}_{D,n_1+j} - E[c_{D,n_1+j} | \FF_{D,n_1}]\right|  = \|\hat{\Sigma}_{21}\hat{V}_{11}(\hat{U}_{11}^{\top}\hat{\Sigma}_{11}\hat{V}_{11})^{-1}\hat{U}_{11}^{\top}\hat{\Sigma}_{12} - A_{21}(V_{11}\Lambda_{11}^{-1}U_{11}^{\top})A_{12}\|_{\infty}\\
        & \quad = \|\hat{A}_{21}\hat{V}_{11}(\hat{U}_{11}^{\top}\hat{\Sigma}_{11}\hat{V}_{11})^{-1}\hat{U}_{11}^{\top}\hat{A}_{12} - A_{21}(V_{11}\Lambda_{11}^{-1}U_{11}^{\top})A_{12}\|_{\infty}\\
            & \quad \leq \|\hat{A}_{21}\hat{V}_{11}((\hat{U}_{11}^{\top}\hat{\Sigma}_{11}\hat{V}_{11})^{-1} - \Lambda_{11}^{-1})\hat{U}_{11}^{\top}\hat{A}_{12}\|_{\infty} + \|(\hat{A}_{21}\hat{V}_{11} - A_{21}V_{11})\Lambda_{11}^{-1}(\hat{U}_{11}^{\top}\hat{A}_{12} - U_{11}^{\top}A_{12})\|_{\infty}\\
            & \qquad\quad+ \|A_{21}V_{11}\Lambda_{11}^{-1}(\hat{U}_{11}^{\top}\hat{A}_{12} - U_{11}^{\top}A_{12})\|_{\infty} + \|(\hat{A}_{21}\hat{V}_{11} - A_{21}V_{11})\Lambda_{11}^{-1} U_{11}^{\top}A_{12})\|_{\infty}
            \\
            & \quad = O_{P}\Bigg(\sqrt{n_{1}D}\|(\hat{U}_{11}^{\top}\hat{\Sigma}_{11}\hat{V}_{11})^{-1} - \Lambda_{11}^{-1}\|_{\max} + \frac{1}{\sqrt{n_1}}\|\hat{A}_{21}\hat{V}_{11} - A_{21}V_{11}\|_{\max} \\
            & \qquad\qquad + \frac{1}{\sqrt{D}}\|\hat{U}_{11}^{\top}\hat{A}_{12} - U_{11}^{\top}A_{12}\|_{\max}\Bigg)\\
            & \quad= O_{P}\Bigg(\rho_{m} +  \frac{1}{m^{\frac{1}{4}}} + \sqrt{\frac{\log D}{n_{1}}} + \frac{\varphi_D}{D}+ \sqrt{\frac{\log n_{1}}{D}} + \frac{\varphi_n}{n_{1}} + \rho_{m} +\frac{\varphi_{D}}{D} + \frac{1}{m^{\frac{1}{4}}} +   \sqrt{\frac{\log D}{n_1}}\\
        & \qquad\qquad\qquad + \rho_{m} +\frac{\varphi_{n}}{n}+ \frac{1}{m^{\frac{1}{4}}} + \sqrt{\frac{\log n}{D}} \Bigg)\\
        & \quad= O_{P}\Bigg( \rho_{m} + \frac{1}{m^{\frac{1}{4}}} +  \frac{\varphi_{D}}{D} + \sqrt{\frac{\log D}{n_1}}+    \frac{\varphi_{n}}{n_{1}}+ \sqrt{\frac{\log n}{D}} \Bigg).
        \end{align*}
	$\square$

	\end{document}